\documentclass[submission,hidelink]{dmtcs-episciences}

\usepackage[round]{natbib}

\usepackage[utf8]{inputenc}
\hypersetup{%
    pdfborder = {0 0 0}
}

\usepackage{subfigure}

\usepackage{amsthm}

\usepackage{stmaryrd}
\usepackage{mathdots}
\usepackage{float}
\usepackage{multirow}
\usepackage{here}
\usepackage{ae,aeguill}
\usepackage{fancyhdr}
\usepackage{amssymb}
\usepackage{wrapfig}
\usepackage{amsmath}
\usepackage{color}
\usepackage{amsfonts}
\usepackage[toc,page]{appendix}
\usepackage{fancybox}
\usepackage{moreverb}
\usepackage{pgf}
\usepackage{tabularx}
\usepackage{tikz}
\usepackage{MnSymbol}

\usetikzlibrary{calc,automata,arrows,chains,matrix,positioning,scopes}

\makeatletter
\newtheorem{theorem}{Theorem}
\newtheorem{proposition}[theorem]{Proposition}
\newtheorem{corollary}[theorem]{Corollary}
\newtheorem{example}{Example}
\newtheorem{definition}[theorem]{Definition}
\newtheorem{lemma}[theorem]{Lemma}

\newcommand{\WF}{K}
\newcommand{\WFp}[1]{\textrm{K}_{#1}}
\newcommand{\WFpp}[1]{\textrm{F}_{#1}}

\newcommand{\Zd}{\Z_d}

\newcommand{\Z}{{\mathbb Z}}     

\renewcommand{\leq}{\leqslant}
\renewcommand{\geq}{\geqslant}
\newcommand{\D}{\mathbf{D}}

\newcommand{\A}{\mathbf{A}}

\newcommand{\equations}[1]{{\llbracket #1\rrbracket}}

\newcommand{\DA}{\mathbf{DA}}

\newcommand{\QV}{\mathbf{QV}}

\newcommand{\QLV}{\mathbf{QLV}}
\newcommand{\LV}{\mathbf{LV}}
\newcommand{\lV}{\mathbf{\ell V}}

\newcommand{\V}{\mathbf{V}}
\newcommand{\gV}{\mathbf{gV}}
\newcommand{\Com}{\mathbf{Com}}
\newcommand{\cd}{\textrm{cd}}
\newcommand{\J}{\mathbf{J}}

\newcommand{\cF}{\mathcal{F}}

\newcommand{\FO}{\mathbf{FO}}
\newcommand{\MSO}{\mathbf{MSO}}
\newcommand{\BS}{\mathbf{{\cal B}\Sigma}}

\newcommand{\Ae}{A_d}
\newcommand{\Reg}{\mathrm{Reg}}
\newcommand{\MOD}{\mathrm{MOD}}
\newcommand{\loc}{\mathrm{LOC}}

\newcommand{\MODV}{\mathbf{MOD}}

\newcommand{\Fsigmod}{\cF[\sigma,\MOD]}
\newcommand{\Fsig}{\cF[\sigma]}

\newcommand{\Fsigmodd}{\cF[\sigma, \MOD^d]}

\newcommand{\rank}{$\text{rank}$}
\newcommand{\inv}{^{-1}}

\newcommand{\Ob}{$\,\textrm{Ob}$}

\newenvironment{itemize2}[1]%
   {%
   \begin{list}{--}%
   {\setlength{\leftmargin}{0cm}%
   }%
   #1%
   }%
   {%
   \end{list}%
   }%

\newenvironment{conditions}
{%
	\begin{list}{\rm (\theenumi)}%
	{\noindent%
		\usecounter{enumi}%
		\setlength{\topsep}{2pt}%
		\setlength{\partopsep}{0pt}%
															 \setlength{\itemsep}{2pt}%
		\setlength{\parsep}{0pt}%
		\setlength{\leftmargin}{2.5em}%
		\setlength{\labelwidth}{1.5em}%
		\setlength{\labelsep}{0.5em}%
		\setlength{\listparindent}{0pt}%
		\setlength{\itemindent}{0pt}%
	}%
}%
{\end{list}}%

\newenvironment{bulitem}
{%
	\begin{list}{\rm $\bullet$}%
	{\noindent%
		\usecounter{enumi}%
		\setlength{\topsep}{2pt}%
		\setlength{\partopsep}{0pt}%
															 \setlength{\itemsep}{2pt}%
		\setlength{\parsep}{0pt}%
		\setlength{\leftmargin}{2.5em}%
		\setlength{\labelwidth}{1.5em}%
		\setlength{\labelsep}{0.5em}%
		\setlength{\listparindent}{0pt}%
		\setlength{\itemindent}{0pt}%
	}%
}%
{\end{list}}%

\newcommand{\Arr}{\textrm{Arr}}

\newcommand{\tvi}{\vrule height 12pt depth 5 pt width 0 pt}

\DeclareFontFamily{U}{mathx}{\hyphenchar\font45}
\DeclareFontShape{U}{mathx}{m}{n}{
      <5> <6> <7> <8> <9> <10>
      <10.95> <12> <14.4> <17.28> <20.74> <24.88>
      mathx10
      }{}
\DeclareSymbolFont{mathx}{U}{mathx}{m}{n}

\newcommand{\tinf}[1]{\mathcal{I}_E(#1)}
\newcommand{\ACOM}{\mathbf{ACom}}

\author[L. Dartois, C. Paperman]{Luc Dartois\affiliationmark{1,2}
  \and Charles Paperman\affiliationmark{3}}
\title[Adding modular predicates]{Adding modular predicates to first-order fragments}
\affiliation{
  LIF, Aix-Marseille Université, France\\
  Université Libre de Bruxelles, Belgium\\
  Warsaw University, Poland}

  	\date{\today}

\begin{document}
\maketitle

\begin{abstract}
We investigate the decidability of the definability problem for fragments of first order logic over finite words enriched with modular predicates.
Our approach aims toward the most generic statements that we could achieve, which successfully covers the quantifier alternation hierarchy of first order logic and some of its fragments.
 We obtain that deciding this problem for each level of the alternation hierarchy of
 both first order logic and its two-variable fragment when equipped with all regular numerical predicates
 is not harder than deciding it for the corresponding level
 equipped with only the linear order and the successor. For two-variable fragments we also treat the case of the signature containing
 only the order
 and modular predicates.

Relying on some recent results, this proves the decidability for each level of the alternation hierarchy of the two-variable first order fragment
while in the case of the first order logic the question remains open for levels greater than two.

The main ingredients of the proofs are syntactic transformations of first order formulas as well as the algebraic framework of finite categories.
%
\end{abstract}

\section{Introduction}

The equivalence between regular languages and automata~\citep{RabinScott59}
 as well as monadic second order
logic~\citep{Buchi60} and finite monoids~\citep{Nerode59} was the start of a
domain of research that is still active today.
	In this article, we are  interested in the logic on finite words, and more precisely the
	question we address is the \emph{definability problem} for fragments of logic.
	Fragments of logic are defined as sets of monadic second order formulas satisfying some restrictions, and are
	equipped with a set of predicates called a \emph{signature}.
	Then the definability problem of a fragment of logic $\cF$ consists in deciding
	 if
	a regular language can be defined by a formula of $\cF$.

	This question has already been considered and solved in many cases where the signature
	contains only the predicate $<$, which denotes the linear order over the positions of the word. For instance,
	a celebrated result by~\cite{Schutzenberger65} and~\cite{MP71}
	gave an effective algebraic characterization of languages definable by first order formulas.
	The decidability has often been achieved through
	algebraic means, showing a deep connection between algebraic and logical properties of
	a given regular language.
	This is the approach privileged in this article.

	We investigate the question of the behaviour of the decidability of some fragments
	when their signature is enriched with \emph{modular predicates}.
	These predicates allow to specify the congruence of the position of a variable modulo an integer.
	They form with the order and the local predicates the set of \emph{regular numerical predicates}.
	These predicates are exactly the formulas of monadic second order logic
	without letter predicates. Intuitively they correspond to the maximal
	class of \emph{numerical predicates} that can enrich the signature of a fragment of $\MSO$, while keeping the definable languages regular.
	This question was already considered in the case of first order logic ($\FO$) by~\cite{Bar92} and one
	of its fragments, the formulas without quantifier alternation, by~\cite{Pel92}.

	The enrichment by regular numerical predicates arose in the context of the
	\emph{Straubing's conjectures}~\citep{Straubing94}. Roughly speaking, these conjectures state that
	deciding the definability  of a regular language in a fragment of enriched logic corresponds to
	 deciding its circuit complexity.
	It is known~\citep{Pel92,Straubing94} that
	an enrichment of the classical fragments by regular numerical predicates is equivalent to an enrichment by the signature
	${[<,+1,\MOD]}$, where $+1$ denotes the \emph{local predicates}
	and $\MOD$ the \emph{modular predicates}.
	A first step toward  the study of fragments of logic with
	these predicates was initiated by~\cite{Straubing85}.
	He obtained  that
	adding the local predicates preserves the decidability for a large number of fragments. As a corollary
	of this work, Straubing obtained that the decidability of the alternation
	hierarchy of first order logic ($\BS_k$) equipped with $[<,+1]$ reduces to the
	decidability of the simpler one $[<]$.
	More recently,~\cite{KL13} proved the decidability of the
	alternation hierarchy of the two-variable first order
	fragment ($\FO^2_k$) equipped with $[<,+1]$ by extending the recent results by~\cite{KS12} and~\cite{KW12} on the
	decidability of this hierarchy with~$[<]$.

	In this context, the case of modular predicates is poorly understood.
	The study of this enrichment
	was first considered for first order logic by~\cite{Bar92}, and had been extended to the first
	level of its alternation hierarchy with the successor predicate by~\cite{Pel92},
	and later without it by~\cite{CPS06b}.
		The enrichment by a finite set of modular predicate was considered by~\cite{EI03}.
	Finally, the authors
	provided a characterization of the two-variable first order logic over the signature~$[<,\MOD]$~\citep{DP13}.
	In this paper, we focus on the enrichment by all regular predicates as well as the question of the enrichment
	by modular predicates only. This latter one surprisingly turns out to be more intricate.

	To study this enrichment in a generic setting, we offer a definition of fragment as a set of formulas satisfying some syntactic properties.
	This allows for some generic proofs instead of a one by one situation.
	The main applications of our theorems are then the \emph{quantifier alternations} hierarchies of the first order logic
	and its two-variable counterpart.
	Our main results state that for both of these hierarchies, the decidability of each level
	equipped with regular numerical predicates reduces to decidability of the same level with the signature $[<,+1]$.
	Then by using
	the recent decidability result of~\cite{KL13}, as well as the decidability of $\BS_2[<]$
	by~\cite{PZ14}, we deduce that the fragments $\FO^2[<,\MOD]$ and $\FO^2_k[\Reg]$, for any positive $k$, as well as $\BS_2[\Reg]$ are decidable.
	Our settings also reproves known results and apply to fragments of first order with small signatures.

	\paragraph{Proofs methods.}
	The proofs of the main results can be decomposed in two major steps.
	 The first part is rather classical and shows that the information given by a finite number of modular predicates can be put into the alphabet and thus we can reduce the problem to a question on the fragment over a bigger alphabet.
	The second part is dedicated to finding a systematic way to select,
	for a given regular language and a fragment, a finite number of
	modular predicates that can serve as witnesses of its definability.
	This is done through the use of the algebraic
	framework of varieties, using two mains approaches.
	The first one uses finite categories and the global of a variety, while the second one
	introduces a new  notion for
	varieties of semigroups that we call the infinitely testable property.
	Under some assumptions, we show that this property
	allows us to find such a witness set for modular predicates.

	\paragraph{Organization of the paper.}
	The next section is dedicated to the basic logical and algebraic definitions, and the main applications of our results to logic are presented in Section~\ref{Section:Main}.
	 Then Section~\ref{Section:Finite} deals with adding a finite number of modular predicates.
	 This is done through an easy reduction to adding predicates modulo a given congruence.
	 In Section~\ref{Section:Delay}, we then deal with the delay problem, which can be quickly stated as computing a finite set of congruences that can serve as a witness for the definability problem of a language.
	 More specifically, we first introduce the framework of categories as an extension of the monoids theory,
	 and use it to prove a delay for differents classes of fragments.
	 In Subsections~\ref{Subsection:Local} and~\ref{Subsection:FiniteRank}, we rely on an algebraic description of the global of a variety, which is a variety of finite categories.
	 Then Subsection~\ref{SsSection:InfTest} solves the delay for a class of fragments satisfying a given algebraic property, the so-called \emph{infinitely testable property}.

\section{Preliminary definitions}\label{Section:Defs}
	\subsection{Languages and Logic}
	We consider the monadic second order logic on finite words $\MSO[<]$ as usual
	(see~\cite{Straubing94} for example).
	We denote by $A$ an \emph{alphabet} and by $a$ a \emph{letter} of $A$.
	A word $u$ over an alphabet $A$ is a set of labelled positions ordered from $0$ to $|u|-1$, where $|u|$ is an integer denoting the \emph{length} of $u$.
	The set of words over $A$ is denoted $A^*$ and a subset $L$ of $A^*$ is called a \emph{language}.
	We also denote by $A^+$ the set of non-empty words.
	A language is said to be \emph{defined} by a formula if it corresponds exactly
	to the set of words that satisfy this formula.
	It is said to be \emph{regular} if it is defined by a $\MSO[<]$ formula.
	When syntactic restrictions are applied to $\MSO[<]$,
	one defines fragments of logic that characterize subclasses of regular languages.
	The most well-known fragment is probably the first order logic,
	whose expressive power was characterized
	thanks to the results of~\cite{MP71} and \cite{Schutzenberger65}.
	The first order logic itself gave birth to its own zoo of fragments.
	These were defined using syntactical restrictions such as limiting the number of variables,
	or by enrichment of its signature.
	A fragment $\cF$ with signature $\sigma$ will  be denoted $\cF[\sigma]$ and will refer to the formulas as well as the class of languages it
	defines.

	We first define the different signatures that will appear through this paper,
	and then formally define the quantifier alternation hierarchies, as they form the main focus of the applications of our theorems.

\paragraph{Signatures.}
	We are interested in regular numerical predicates, which are numerical predicates that can only define regular languages.
	Simultaneously,~\cite{Straubing94} and~\cite{Pel92}  defined three sets of regular numerical predicates that can be used as a base for all the regular numerical predicates.
	The first set is the singleton order $\{<\}$ which is a binary predicate corresponding to the natural order on the positions of the input word.
	The second set is $\{\mathbf{min},\mathbf{max},\mathbf{S_k}\}$ and is called the \emph{local predicates}.
	The predicates $\mathbf{min}$ and $\mathbf{max}$ are unary predicates that are satisfied respectively on the first and last positions.
	The predicate $\mathbf{S_k}$, the $k^{\text{th}}$-\emph{successor}, is a binary predicate satisfied if the second variable quantifies the
	$k^{\text{th}}$-successor of the first one.

	\begin{example}
		The formula $\exists x\exists y \ \mathbf{min}(x)\land \mathbf{S}(x,y) \land \mathbf a(x) \land \mathbf a(y)$
		defines the regular language $aaA^*$.
	\end{example}
	We alternatively use the \emph{descriptive local predicates}.
	These predicates are of the form
	$\mathbf{a}(x+k)$ (resp. $\mathbf{a}(\min+k)$, $\mathbf{a}(\max-k)$) for $k\geq 0$, holding if the position at $x+k$ (resp. $min+k$, $max-k$) is labelled by an $a$.
	\begin{example}
		The previous formula can be rewrite by the following quantifier-free formula: $\mathbf{a}(\min)\land \mathbf{a}(\min+1)$.
	\end{example}
	Most of the time, both \emph{descriptive} and classical local predicates provides the same expressive power. However the descriptive
	predicates are proved to be more convenient for abstract fragments since they don't bound two variables together. For the sake of
	simplicity we will denote in the following by $+1$ this class of descriptive local predicates. This notation is justified thanks to
	the close relation between descriptive local predicates and the successor function.
	Also note that the presence or absence of the equality predicate is important since $\FO[+1]$ is strictly less expressive than
	$\FO[=,+1]$.

	Finally, we define, for each positive integer $d$, the \emph{modular predicates on $d$}, denoted $\MOD^d$, as the set, for $i<d$, of predicates
$\mathbf{MOD_i^d}(x)$ which are unary predicates satisfied if the position quantified by $x$ is congruent to $i$ modulo $d$, and the predicates
	$\mathbf{D_i^d}$ which are constants holding if the length of the input word is congruent to $i\bmod d$.
	We denote by $\MOD$ the union of the classes $\MOD^d$, for any positive $d$.
	\begin{example}
		The language $(A^2)^*aA^*$ is defined by the formula:
		$\exists x \ \mathbf a(x) \land\ \mathbf{MOD_0^2}(x)\enspace.$
	\end{example}

The signatures that we will consider for our fragments are unions of these three sets of regular numerical predicates, and will always contain the letter predicates.
Abusing notations, we will also write $\Reg=\{<\}\cup +1 \cup\MOD$.

\paragraph{Fragments.}
\
	A \emph{fragment of logic} $\cF[\sigma]$ with signature $\sigma$ is a set of closed formulas of $\MSO[\sigma]$
	that contains the quantifier-free formulas and that is
	closed
	under the following operations :
	\begin{description}
	\item[\textit{Conjunction}] If $\varphi$ and $\psi$ are formulas of $\cF$,
				then $\varphi\wedge \psi$ is also a formula of $\cF$.
	\item[\textit{Disjunction}] If $\varphi$ and $\psi$ are formulas of $\cF$,
				then $\varphi\vee \psi$ is also a formula of $\cF$.
	\item[\textit{Quantifier-free substitutions}] If $\varphi$ is a formula of $\cF$ and $\psi(x_1,\ldots,x_n)$ a quantifier-free subformula of $\varphi$ with free variables $x_1,\ldots,x_n$, then any formula obtained by replacing $\psi(x_1,\ldots,x_n)$ by another quantifier-free formula with the same set of free variables is also in $\cF$.
	\end{description}

	If $\cF[\sigma]$ is a fragment of logic and $\sigma'$ is a class of predicates, then the \emph{enrichment} of $\cF[\sigma]$ by $\sigma'$ is denoted by $\cF[\sigma,\sigma']$ and corresponds to the closure of $\cF[\sigma]$ under the quantifier-free substitutions, where predicates range over the signature $\sigma\cup \sigma'$.
	As a closed formula defines a language, a fragment of logic defines a class of languages. Abusing notations, we will denote by $\cF[\sigma]$ a
	fragment of logic, as well as the class of languages it recognizes.
	It is worth noting that~\cite{KL12} defined another notion of fragment of logic as sets of
	formulas closed under some syntactical substitutions ensuring algebraic characterisation of the fragment.

	The fragment $\FO^2$ is the subclass of formulas of $\FO$ using only two symbols of variables which can be reused (see Example~\ref{FO2:ex}).
	Here, the class of languages defined by $\FO^2[<]$ is strictly contained in $\FO^2[<,+1]$ and $\FO^2[<,\MOD]$ (see~\cite{TW98,DP13}).
	\begin{example}\label{FO2:ex}
	The language $A^*aA^*bA^*aA^*$ can be described by the first order formula
	$$\exists x \exists y \exists z\ x <y < z \land \mathbf{a}(x)\land \mathbf{b}(y)\land \mathbf{a}(z)\enspace.$$
	This formula uses three variables $x,y$ and $z$. However, by reusing $x$ we get an equivalent formula that uses only two variables:
	\begin{align}
	\exists x\ \mathbf{a}(x) \land \Big(\exists y\  x < y  \land \mathbf{b}(y)\land \big(\exists x \ y<x\land \mathbf{a}(x)\big)\Big)\enspace.\tag{a}\label{eq1}
	\end{align}
	\end{example}

\paragraph{Alternation hierarchies.}
	Given a first order formula, one can compute a prenex normal form using the De Morgan's laws.
	We define the \emph{quantifier alternation depth} of a formula as the number of blocks  of quantifiers $\forall$ and $\exists$ in its prenex normal form.
	For example, the formula $\exists x\exists y \forall z\ x\!<\!z\!<\!y \land \mathbf{a}(x)\land \mathbf{a}(y)\land \mathbf{c}(z)$ has a quantifier depth of $2$.
	It describes the language $A^*ac^*aA^*$.
	Then given a signature $\sigma$ and a positive integer $k$, we denote by $\BS_k[\sigma]$ the set of prenex normal formulas of $\FO[\sigma]$ whose quantifier depth is smaller or equal to $k$.
 They form the levels of the \emph{quantifier alternation hierarchy} over $\FO[\sigma]$.

	When $\sigma$ is reduced to $\{<\}$, this hierarchy is called the Straubing-Th\'erien hierarchy~\citep{Straubing81,Th81}.
	Only the first~\citep{SI75} and second~\citep{PZ14} levels are known to be decidable.
	For ${\sigma=\{<\}\cup +1}$, this hierarchy is called the Dot-Depth hierarchy~\citep{BC71}.
	The decidability of each level reduces to the decidability of the corresponding
	level of the Straubing-Th\'erien hierarchy~\citep{Straubing85}.
	In both cases, the hierarchies are known to be strict, and cover all Star-Free languages.
	In this article, we also consider the alternation hierarchy of $\FO^2$.
	To define formally the number of alternations of a formula,
	we cannot rely on the prenex normal form since the construction increases the number of variables.
	In particular, remark that $\FO^2[<]$ is equivalent to $\mathbf{\Sigma_2}[<]\cap \mathbf{\Pi}_2[<]$ which
	is a subclass of $\BS_2[<]$~\citep{DGK08}.
	 That said, the number
	of alternations is still a relevant parameter that could be defined as follows:
	Consider the parse tree naturally associated to a formula. For instance,
	\eqref{eq1} has $\exists$ as a root  and the atomic formulas as the leaves.
	In a two-variable first order formula we count the maximal number of alternations appearing on a branch, i.e. between the
	root and a leaf,
	once the negations have been pushed on to the leaves. A more precise definition
	can be found in~\cite{IW09}. We denote by $\FO^2_k[\sigma]$
	the formulas of $\FO^2[\sigma]$ that have at most $k-1$ quantifier alternations.
	The hierarchy induced by $\FO^2_k[<]$ is known to be strict~\citep{IW09} and its definability
	problem is decidable~\citep{KS12,KW12}.
	Note that the hierarchy $\FO^2_k[<,+1]$ is also known to be decidable~\citep{KL13}.

	\textbf{Remark:} The classes of formulas $\FO$ and $\FO^2$
	as well as each level of the alternation hierarchies are fragments of $\MSO$ as defined previously.

	\subsection{Varieties of languages, monoids and semigroups}
We quickly present here the fundamental notions used by the article and refer the reader to the book of~\cite{Pin97a} for a detailed approach.
A (finite) \emph{semigroup} is a finite set equipped with an associative internal law.
A semigroup with a neutral element for this law is called a \emph{monoid}.
Recall that a semigroup $S$ \emph{divides} another semigroup $T$ if $S$ is a quotient of a
subsemigroup of $T$.
This defines a partial order on finite semigroups.
Given a finite semigroup $S$, an element $e$ of $S$ is \emph{idempotent} if $ee=e$.
We denote by $E(S)$ the set of idempotents of $S$.
For any element $x$ of $S$, there exists a positive integer $n$ such that $x^n$ is idempotent. We call this element the \emph{idempotent power} of $x$ and denote it by $x^\omega$.
One can check that the application $x\to x^\omega$ is well defined.

A semigroup $S$ recognizes a language $L$ over an alphabet $A$ via a \emph{morphism}
${\eta:A^+\to S}$.
Given a regular language $L$, we can compute its \emph{syntactic semigroup} as the smallest semigroup that recognizes $L$, in the sense of division.
A subset $T$ of $S$ is an \emph{ideal} if the sets $TS$ and $ST$ are both included in $T$.
A (pseudo-)\emph{variety} of semigroups (resp. monoids) is a non empty class of finite semigroups (resp. monoids) closed under
division and finite product.
Finally, a local monoid of $S$ is a monoid of the form $eSe$ where $e$ is an idempotent of $S$.

A fragment of logic is \emph{characterized} by a variety if they recognize the same languages.
By extension, a variety $\V$ will also refer to the class of languages it recognizes.
The most famous example is the equality $\FO[<]=\A$~\citep{MP71,Schutzenberger65}, where $\A$ denotes the class of aperiodic semigroups, which are
finite semigroups that are not divided by any group.
As for $\FO[<]$, the definability problem for a fragment of logic has often been solved thanks to an algebraic characterization (\cite{SI75,Th81,TW98} for example).
This decidability is sometimes obtained through \emph{profinite equations}.
We refer the reader to~\cite{PIN09} for a survey on the profinite background. 
The algebraic characterisations of most the fragments that we consider are given in Figure~\ref{Tableau:Alg}.
\begin{figure}[h]
\begin{tabular}{|c|c|c|}
\hline\tvi
Fragment & Variety & Equations \\
\hline\tvi
$\FO[<]$ & $\A$  & $x^\omega=x^{\omega+1}$ \\
\hline\tvi
$\FO[=]$ & $\ACOM$ & $x^\omega=x^{\omega+1}$, $xy=yx$ \\
\hline\tvi
$\FO^2[<] $& $\DA$ & $(xy)^\omega= (xy)^\omega x (xy)^\omega$ \\
\hline\tvi
$\FO^1[\emptyset]$ & $\J_1$ & $x^2=x$, $xy=yx$ \\
\hline\tvi
$\mathbf{\BS}_1[<]$ & $\J$ & $y(xy)^\omega=(xy)^\omega=(xy)^\omega x$\\
\hline\tvi
$\FO^2_k[<]$ & $\V_k$ & See Example~\ref{FinRank-Ex} \\
\hline

\end{tabular}
\centering
\caption{Algebraic characterisations}\label{Tableau:Alg}
\end{figure}

\paragraph{Stability index.}
One important tool to study modular predicates is the stability index.
For a monoid morphism $\varphi : A^* \to M$, the set $\varphi(A)$ is an element of the powerset
monoid of $M$.
As such it has an idempotent power. The \emph{stability index} of a morphism
is the 	least positive integer $s$
such that $\varphi(A^s)=\varphi(A^{2s})$. This set forms a subsemigroup called
the \emph{stable semigroup} of
$\varphi$.
The set $\varphi((A^s)^*)$ is called the \emph{stable monoid}
of $\varphi$.
The stable monoid (resp. semigroup) of a regular language is the stable monoid (resp. semigroup) of its syntactic morphism.

\section{Adding finitely many modular predicates}\label{Section:Finite}

We consider here the question of adding the modular predicates associated to a finite set of congruences.
First, let us remark that if
$\cF[\sigma]$ is a fragment of logic and $d$ and $p$ are two positive integers,
then $\cF[\sigma,\MOD^d,\MOD^p]\subseteq\cF[\sigma,\MOD^{dp}]$. This can be proved by some basic arithmetic reasoning and some quantifier-free substitutions.
Then as a formula only uses a finite number of modular predicates, for any language of $\cF[\sigma,\MOD]$, there exists an integer $d$ such that it belongs to $\cF[\sigma,\MOD^d]$.
The consequence is that adding a finite set of modular predicates is equivalent to adding the predicates relating to one specific congruence.
The remainder of this section deals with this question.

\subsection{Alphabet enriched by modular counting}

In order to deal with modular predicates, we now define enriched modular alphabets. These notions come naturally in the context of \emph{wreath product} and instantiated for instance in~\cite{DP13,DP15}.
We now fix a positive integer $d$ and an alphabet $A$. Let $\Zd$ be the cyclic group of order $d$.
	\begin{definition}[Enriched alphabet]
		We call the set $\Ae=A\times \Zd$
		the \emph{enriched alphabet} of $A$,
		and we denote by $\pi_d:\Ae^*\to A^*$ the projection defined by $\pi_d(a,i)=a$
		for each $(a,i)\in \Ae$.
		For example,
		the word $(a,2)(b,1)(b,2)(a,0)$ is an enriched word of $abba$ for $d=3$.
		We say that $abba$ is the \emph{underlying word} of $(a,2)(b,1)(b,2)(a,0)$.
	\end{definition}
	\begin{definition}[Well-formed words]
		A word $(a_0,i_0)(a_1,i_1)\dotsm(a_n,i_n)$ of $\Ae^*$ is \emph{well-formed} if
		for $0\leq j\leq n$, $i_j = j \bmod d$. We denote by $\WF_d$ the set of
		all well-formed words of $\Ae^*$.
		We also note $\Ae(i,j)$ the set of well-formed factor such that the first letter is labelled by $i$ and the last by $j$.

		For any $i<d$, let $\alpha_d^i: A^*\to \Ae^*$ be the function defined
		for any word $u=a_0a_1\dotsm a_n\in A^*$
		by $\alpha_d^i(u)=(a_0,i)
		(a_1,i+1\bmod d)\dotsm(a_n,i+n\bmod{d})$.
		We simply denote $\alpha_d^0$ by $\alpha_d$ and the word $\alpha_d(u)$ is called the \emph{well-formed word attached} to u.
	\end{definition}

		Note that the restriction of $\pi_d$ to the set of well-formed words is one-to-one.
		For instance, the enriched word $(a,0)(b,1)(b,2)(a,0)$ is a well-formed word for $d=3$.
		It is the unique well-formed word having the word
	$abba$ as underlying word. Finally, given a language $L$, we write $L_d=\pi_d^{-1}(L)\cap \WFp{d}$.

\subsection{A first transfer theorem}
Using the enriched alphabet and the well-formed words, the next theorem links a fragment with its enrichment by congruences modulo one integer.
It transfers the expressiveness of modular predicates to the alphabet.
An aware reader could notice that it is very similar to the wreath product principle of varieties. It is in fact not a coincidence, since
this operation matches with a wreath product by the \emph{length-multiplying} variety ${\MODV}$ (see~\cite{CPS06b} for more details).
\begin{theorem}\label{semimod}
			Let $\cF[\sigma]$ be a fragment of logic,
			$L$ a regular language and $d$ a positive integer.
			Then the following properties are equivalent:
			\begin{conditions}
				\item\label{thmsemimod:1} $L$ is definable by a formula of $\cF[\sigma,\MOD^d]$,

				\item\label{thmsemimod:3} there exists
				      some languages $L_0,\ldots,L_{d-1}$ of
				      $\cF[\sigma]$ over $\Ae^*$ such that:
				      \begin{equation}\label{equationLd}
				      L =\bigcup_{i=0}^{d-1}\big( (A^d)^*A^i\cap \pi_d(L_i\cap \WF_d)\big)
				      	      \tag{a}
				      \end{equation}
			\end{conditions}
		\end{theorem}
		To prove the result,
 		we need an auxiliary result which gives a decomposition of the language defined by a formula into smaller pieces.

	\begin{lemma}\label{formNorm}
	Let $\Fsigmod$ be a fragment of logic and $\varphi$ a formula of $\Fsigmodd$.
	Then there exists $d$ formulas $\psi_i$ of $\Fsigmodd$ that do not contain any predicate $D_j^d$
	and such that $$\varphi\equiv \bigvee_{i=0}^{d-1}(\psi_i\wedge D_i^d).$$
	Moreover, we have:  
	 $$L(\varphi)=\bigcup_{i = 0}^{d-1}  \big((A^d)^*A^i\cap L(\psi_i)\big).$$
	\end{lemma}
	\begin{proof}
	For $i<d$, we define the formula $\psi_i$ to be the formula $\varphi$ where we replaced
		every predicate $D_i^d$ by \emph{true} and every $D_j^d$ with $j\neq i$ by \emph{false}.
		One should notice that, by definition of a fragment, the formulas
		$\psi_i$ are in $\Fsigmodd$.
		We can conclude the proof since the formula $(D^d_i)$
		recognizes the language $(A^d)^*A^i$.
	\end{proof}
%

		\begin{proof}[of theorem~\ref{semimod}]
		Let $\varphi$ be a formula of $\Fsigmod$. Then $\varphi$ belongs to $\Fsigmodd$ for some $d>0$.
		Using Lemma \ref{formNorm}, we know it is sufficient to consider a formula
		$\varphi$ without any length predicate.
		We transform it into a formula $\psi$ by doing the following transformation:
		\begin{align*}
			 \MOD_i^d(x) \text{ is replaced by } \bigvee_{a\in A} \mathbf{(a,i)}(x),\\
			 	\mathbf a(x)\text{  is replaced by }\bigvee_{0\leq i<d} \mathbf{(a,i)}(x).
		\end{align*}
		The resulting formula
		$\psi$ is in $\Fsig(A_d^*)$ and $L(\varphi)=\pi_d(L(\psi)\cap \WF_d)$.
		Conversely, we transform a formula $\psi$ of $\Fsig(A_d^*)$ into a formula $\varphi$ of $\Fsigmodd$ by
		replacing every predicate $\mathbf{(a,i)}(x)$ in $\psi$ by $\mathbf a(x)\wedge \MOD_i^d(x)$.
		We also get $L(\varphi)=\pi_d(L(\psi)\cap \WF_d)$.
	\end{proof}

	The previous theorem provides a semantic counterpart to the action of adding modular predicates to a fragment of logic. In the case where
	the fragment is \emph{expressive enough}, this counterpart provides a transfer of decidability, as stated in the next
	corollary.
	\begin{corollary}[The transfer result]\label{Cor:transfer}
		Let $\cF[\sigma]$ be a fragment of logic. If $\cF[\sigma]$ is decidable and if both $\WF_d$ and $\textbf{max}$
		are
		definable in
		$\cF[\sigma]$, then
		$\cF[\sigma,\MOD^d]$ is decidable.
	\end{corollary}
	\begin{proof}
	The result comes from the fact that if $\textbf{max}$ is definable, then using modular predicates the languages $(A^d)^*A^i$ are definable in $\cF[\sigma,\MOD^d]$.
	If furthermore we can define the language of well-formed words, then item~\ref{thmsemimod:3} of Theorem~\ref{semimod} is equivalent to the language $L_d$ being definable in $\cF[\sigma]$ over the enriched alphabet.
	This language being computable from $L$, we get decidability.
	\end{proof}
	\textbf{Remark:} Corollary~\ref{Cor:transfer} applies to fragments $\BS_k[\sigma]$, $\FO[\sigma]$, when $k\geq 2$ and $\sigma$ contains either $+1$ or the order. It
	also applies to fragments $\BS_1[\sigma]$, $\FO^2_k[\sigma]$,  or $\FO^2[\sigma]$ when
	$+1$ is contained in $\sigma$.

\section{Main results}\label{Section:Main}

As stated in the previous section, any language defined by a fragment with modular predicates can be done so with a formula using only congruences to one specific integer.
In fact, there exists an infinite number of such witnesses.
The remaining of the article is dedicated to the problem of deciding one witness, given a language.
We call it the \emph{delay problem} and can be explicitly stated as follows:\\
\noindent\textbf{The delay question:}
Given a regular language $L$ and a fragment $\cF[\sigma]$, is it possible to compute an integer $d$ such that $L$ belongs to $\cF[\sigma,\MOD]$ \emph{if, and only if}, it belongs to $\cF[\sigma,\MOD^{d}]$?\\

\noindent
Remark that such an integer $d$ could depend of $L$ and $\cF[\sigma]$. The denomination stems from the Delay Theorem of~\cite{Straubing85} that solves a similar question for the enrichment
by the
successor predicate.
Section~\ref{Section:Delay} is devoted to solve the delay problem for different classes of varieties.
It relies heavily on algebraic notions, in particular the framework of categories.
We present here the main applications to fragments of logic, which are summed up in Figure~\ref{TableauFinal}.
This figure does not include decidability of the smaller fragments of $\FO$ equipped with modular predicates: $\FO[\MOD]$ (Theorem~\ref{Main:local}), $\FO[+1,\MOD]$ (Theorem~\ref{Main:local2}), $\FO[=,\MOD]$ (Theorem~\ref{Main:finiterank}) and $\FO[=,+1,\MOD]$ (Theorem~\ref{Main:InfTest}).
\newcommand{\newresult}{\multirow{2}{60pt}{\textbf{New result}}}
\begin{figure}

\centering

\begin{tabular}{|c|c|c|c|}
\hline
\tvi & $[<]$ & $[<,\MOD]$ & $[\Reg]$ \\
\hline 
\tvi \multirow{2}{60pt}{ $\BS_1=\FO^2_1$} & \cite{SI75} & \multirow{2}{110pt}{\cite{CPS06b}}& \multirow{2}{100pt}{ \cite{MPT00}} \\
&\cite{TO82}&&\\
 \hline
\tvi \multirow{2}{60pt}{ $\FO^2_k$} & \cite{KS12} &  \newresult &  \newresult\\
\tvi & \cite{KW12} & & \\
 \hline
 
 \tvi \multirow{2}{60pt}{\hfil $\FO^2$} &  \multirow{2}{120pt}{\cite{TW98}} & \multirow{2}{130pt}{\cite{DP13}} &\newresult \\
 \tvi & & & \\
 \hline
 
  \tvi \multirow{2}{60pt}{\hfil $\BS_2$} &  \multirow{2}{110pt}{\cite{PZ14}} & \newresult & \newresult \\
  \tvi & & & \\
 \hline
 
 \tvi \multirow{2}{60pt}{\hfil $\BS_k$} & \multirow{2}{40pt}{Open }& Reduces to $[<]$ & Reduces to $[<]$\\
\tvi & & \textbf{New result}& \textbf{New result} \\
\hline 

 \tvi \multirow{2}{60pt}{\hfil $\FO$} & \cite{MP71} &   \multirow{2}{75pt}{\cite{Straubing94}}  &  \multirow{2}{110pt}{\cite{Bar92}} \\
\tvi & \cite{Schutzenberger65} &  & \\
\hline 
\end{tabular}
\caption{Decidability results of first-order fragments}\label{TableauFinal}
\end{figure}

The first decidability results comes from the local property.
Although it does not bring many new results, mainly reproving~\cite{Bar92} and~\cite{DP13}, it gives a unified proof for these fragments.
\emph{Local} varieties have a particular role in the previous work of Straubing, where they are identified as varieties that \emph{behave}
gently compared toward $+1$. In the context of modular predicates, they also have this good property that allows us to state a
fairly generic statement under this assumption. A formal definition of locality can be found in Section~\ref{Subsection:Local}.
\begin{theorem}[Local case, for monoids varieties]\label{Main:local}
	Let $\cF[\sigma]$ be a fragment equivalent to a local variety $\V$.
Now let $L$ be a regular language and $s$ its stability index, then the following statements are equivalent.
\begin{itemize}
	\item  $L$ belongs to $\cF[\sigma,\MOD]$.
	\item  $L$ belongs to $\cF[\sigma,\MOD^s]$.
	\item  the stable monoid of $L$ belongs to $\V$.
\end{itemize}
Furthermore, if $\cF[\sigma]$ is decidable, then so is $\cF[\sigma,\MOD]$.
\end{theorem}
Example of interest includes  $\FO^1[\emptyset]$, ${\FO^2[<]}$ or $\FO[<]$, which are equivalent to $\J_1$, $\DA$ and $\A$ respectively.
The locality of $\J_1$ and $\A$ can be found in the article of~\cite{Tilson}, the locality of $\DA$ is slightly more intricate (see~\cite{Almeida96}).

When the initial variety is local, we can nest our approach with the one with the successor predicates. It is no longer needed to
use the intricate framework of categories since in this case, we can apply Corollary~\ref{Cor:transfer} to slightly simplify the question.
\begin{theorem}[Local case, for semigroups varieties]\label{Main:local2}
Let $\cF[\sigma]$ be a fragment corresponding to a local variety $\V$
Now let $L$ be a regular language and $s$ its stability index, then the following statements are equivalent.
\begin{itemize}
	\item  $L$ belongs to $\cF[\sigma,+1, \MOD]$.
	\item  $L$ belongs to $\cF[\sigma,+1, \MOD^s]$.
	\item  the local monoids of the stable semigroups belongs to $\V$.
\end{itemize}
Furthermore, if $\cF[\sigma]$ is decidable, then so is $\cF[\sigma,+1,\MOD]$.
\end{theorem}
This theorem is a consequence of Proposition~\ref{prop:QLV}.
Note that both $\FO[+1,\MOD]$ and $\FO^2[<,+1,\MOD]$ fall into the scope of this theorem. In the case of full first order
logic, the successor predicate being definable with the order, the expressiveness remains unchanged. The reduction to logic of these results can be found in Subsection~\ref{SsSection:InfTest} for Theorem~\ref{Main:local2} and Subsection~\ref{Subsection:Local} for Theorem~\ref{Main:local}. Note that this provides decidability.

A generalized approach of the previous results brings fresh ones, although we fail to obtain a delay independent from the fragment.
We need to assume some properties on the \emph{varieties of categories} generated by the initial variety. In particular, we assume
that the \emph{path-equations} of the so called \emph{global} of a variety use a bounded number of vertices. Under this assumption we successfully compute a delay.
\begin{theorem}[Finite rank case]\label{Main:finiterank}
Let $\cF[\sigma]$ be a fragment corresponding to a variety $\V$ of rank $k$.
Now let $L$ be a regular language and $s$ its stability index, then the following statements are equivalent.
\begin{itemize}
	\item  $L$ belongs to $\cF[\sigma, \MOD]$.
	\item  $L$ belongs to $\cF[\sigma,\MOD^{ks}]$.
\end{itemize}
Furthermore, if $\cF[\sigma,+1]$ is decidable, then so is $\cF[\sigma,\MOD]$.
\end{theorem}
Example of application of this theorem include $\FO[=]$ which is known to be equivalent to the variety of rank $2$ of aperiodic and
commutative monoids, as well as the alternation hierarchy of $\FO^2[<]$ whose $k^\text{th}$ level is of rank $2k$. This approach is detailed in Section~\ref{Subsection:FiniteRank}. In those cases, this last theorem also provides decidability by reducing to decidability of
the fragment with the successor predicate.

Finally, the next theorem provides a delay for all fragments containing the successor predicates. In particular,
it reduces the decidability of $\cF[\Reg]$ to the decidability of $\cF[<,+1]$ providing decidability for the fragment
$\FO^2_k[\Reg]$ and a reduction of the decidability of $\BS_k[\Reg]$ to the decidability of $\BS_k[<,+1]$, which itself reduces
to decidability of $\BS_k[<]$ thanks to~\cite{Straubing85}.
\begin{theorem}[Infinitely testable case]\label{Main:InfTest}
Let $\cF[\sigma]$ be a fragment corresponding to a variety $\V$ which is not a variety of groups.
Now let $L$ be a regular language and $s$ its stability index, then the following statements are equivalent.
\begin{itemize}
	\item  $L$ belongs to $\cF[\sigma,+1,\MOD]$.
	\item  $L$ belongs to $\cF[\sigma,+1,\MOD^{s}]$.
\end{itemize}
Furthermore, if $\cF[\sigma,+1]$ is decidable, then so does $\cF[\sigma,+1,\MOD]$.
\end{theorem}
The condition that $\cF[\sigma]$ is not equivalent to a group variety is necessary to apply the simplification of Corollary~\ref{Cor:transfer}.
However, in the case where $\cF[\sigma]$ is indeed a variety of groups, then both $\cF[\sigma,+1]$ and $\cF[\sigma,\MOD]$ are decidable
since varieties of groups are known to be local as variety of monoids but seems intricate when both $+1$ and $\MOD$ are in the signature since groups are not local as varieties of semigroups (for instance see book~\cite[page 104]{RS09}).
The proof of this last theorem is given in Subsection~\ref{SsSection:InfTest}.

\section{Solving the Delay problem}\label{Section:Delay}

This section is devoted to solve the delay question for different classes of varieties.

We first present the framework of finite categories as well as some known results, and use it to
reduce the combinatoric characterisation of Theorem~\ref{semimod} to the decidability of the global of a variety, an algebraic notion from the framework of finite categories.

The remainder of the section then uses this characterisation to solve the delay question for different classes of varieties.
The first case is the simplest one of local varieties, where we get a clear characterisation of $\Fsigmod$.
The second case, the finite rank, is a generalisation of the local case, where an algebraic characterisation of the global is known.
Finally, the last case solves the delay for a class of varieties where little is known about the global.
It is the class of varieties of semigroups \emph{expressive enough} and satisfying an extra property:
the \emph{infinitely testable property}, which is a new notion.
%

\subsection{A derived category theorem}
\paragraph{Finite categories: a short introduction.}

		In this section, we present the theory of finite categories, as an extension of finite monoids.
		Informally, a category can be seen as a partial monoid where only some products are allowed.
		Nonetheless, notions from monoids can be correctly lifted, and we will consider varieties of categories.
		The framework of variety of categories has been successful to obtain algebraic characterizations of
		\emph{wreath products} of varieties~\citep{Tilson}.
		For example, the enrichment by modular predicates can be seen as a wreath product by a variety of morphisms.
		This comes from an adapted version of the Wreath Product Principle, that is evoked by~\cite{CPS06b}.
		We chose not to focus on this, since it would require to introduce additional definitions and proofs that are not necessary and would burden the article.

		\emph{A graph $X$} is a set of \emph{objects} denoted $\Ob(X)$
		such that for any couple of objects $(x,y)\in \Ob(X)$, we associate a set $X(x,y)$ of \emph{arrows} from $x$ to $y$.
		Two arrows	$e,f$ are coterminal if there exists $x,y\in \Ob(X)$ such that $e,f\in X(x,y)$. They are consecutive if there exists
		$x,y,z\in \Ob(X)$ such that $e\in X(x,y)$ and $f\in X(y,z)$.  An arrow $e$ is a loop from
		$x$ if $e\in X(x,x)$. A  composition law associates to each pairs of consecutive arrows, $e,f$ an arrow $ef$.
		This law is said to be associative if for any consecutive arrows $e,f,g$ we have $(ef)g=e(fg)$.

		A \emph{category} $C$ is a graph
		with an associative composition law and containing
		for each object $x$ an identity denoted $1_x$.
		Thus the set of loops around a given object, equipped with the composition law, forms a monoid,
		called the \emph{local monoid} of that object.
		Note that the terminology of local monoids of a category clashes with the terminology of local monoids of a semigroup.
		In fact, the two coincide when we consider the idempotent category of a semigroup, which is defined later.

		Here we only consider categories as
		a generalization of finite monoids, since a monoid can be viewed as a
		one-object category.
		A \emph{morphism of categories} $\eta:C\to D$ is an application
		$\eta:\Ob(C)\to \Ob(D)$ and for each pairs
		of object $(x,y)\in \Ob(C)$, an application $\eta: C(x,y) \to D(\eta(x),\eta(y))$ such that
		\begin{conditions}
			\item for any consecutive arrows $e,f$ we have $\eta(ef) = \eta(e)\eta(f)$,
			\item for any $x\in \Ob(C)$, $\eta(1_x)=1_{\eta(x)}$.
		\end{conditions}

		A \emph{division} of categories $\tau:C\to D$
		is given by a mapping $\tau:\Ob(C)\to \Ob(D)$,
		and for each pair of objects $e$ and $f$, by a relation $\tau: C(e,f)\to D(\tau(e),\tau(f))$
		such that
			\begin{conditions}
			\item $\tau(x)\tau(y)\subseteq\tau(xy)$ for consecutive arrows $x,y$,
			\item $\tau(x)\neq \emptyset$ for any arrow $x$,
			\item $1_{\tau(e)}\in\tau(1_e)$ for any object $e$ of $C$.
			\end{conditions}
	We remark that the inverse of an onto morphism of categories is a division of a categories (but the converse is not
	true).
	Then a \emph{variety of categories} is a class of categories closed under direct product and division.

	\begin{definition}
		Given a variety of monoids $\V$, the \emph{global} of $\V$, denoted $\gV$, is the class of all categories that divide a monoid of $\V$, when seen as a one-object category.
	\end{definition}

	\noindent\textbf{Remark:} Since the division of categories is a partial order and a variety is closed under product, the class of categories $\gV$ is closed by division and by product, and it is therefore a variety of categories.

	\begin{definition}[Consolidated semigroup, consolidated stamp]
			Let $C$ be a finite category and $\Arr(C)$
			the set of arrows of $C$. We denote by $C_\cd$ the semigroup
			defined on the set
			$$E=\Arr(C) \cup \{0\}$$
			with   for any $x\in E$, $0x = x0 = 0$, and for $x,y\in \Arr(C)$,
			$$x.y=\begin{cases}
					xy\text{ if }x\text{ and }y\text{ are consecutives arrows,}\\
					0\text{ otherwise}.
					\end{cases}	$$
	\end{definition}
	The following proposition is a well-known result stating that the membership of a category
	in $\gV$ reduces to the membership of $\V$ is the variety is \emph{expressive enough}. This is a
	\emph{category} version of Corollary~\ref{Cor:transfer} which means that the membership of a language to an expressive enough fragment enriched with a finite set of modular predicates reduces to the membership of a different language to the fragment without them.
	\paragraph{Background: the local predicates and derived category for locally testable language.}
	\
		In this section, we recall some known results that we will be using in the remainder of the article and give some intuitions about their significance.
		We first give the definition of the derived category for definite languages and provide
		the delay theorem of~\cite{Straubing85} as well as its improvement
		by~\cite{Tilson}.
			Let $S$ be a semigroup, $n$ an integer and $\eta:S^+\to S$ the canonical semigroup morphism
			of $S$.
			The \emph{$n$-derived category} of $S$ with respect to definite languages, denoted $D_n(S)$, is
			the category with $S^{\leq n}$ as set of objects,
			and the arrows from $u$ to $v$ are the elements $s$ of $S$ such
			that there exists a word $w\in S^+$ that $\eta(w) = s$ and  the
			suffix of size $n$ of $uw$ is equal to $v$.
			The $n$-derived category with respect to definite languages, of a regular language $L$,
			denoted $D_n(L)$, is the category $D_n(\eta_L(A^+))$.
		Finally we also introduce the \emph{idempotents' category} of a semigroup $S$, denoted by
		$S_E$ and defined by~\cite{Tilson} as follows. Its set of objects are
		the idempotents of $S$. And for $e$ and $f$ two idempotents, we set $S_E(e,f) = eSf$.
		We do not recall the definition of the wreath product of a variety $\V$ by $\D$, denoted by
		$\V*\D$.
		However, as our only use of this product is given by the following theorem, an unfamiliar reader can take
		the following theorem as a definition.

		\begin{theorem}[Delay theorem for definite languages]\label{thm:delaydefinite}
			Let $\V$ be a variety and $S$ a semigroup. The following conditions are equivalent.
			\begin{conditions}
				\item The semigroup $S$ belongs to $\V*\D$.
				\item There exists an integer $n$ such that $D_n(S)$ belongs to $\gV$.
				\item For $n=|S|$, $D_n(S)$ belongs to $\gV$.
				\item The category $S_E$ belongs to $\gV$.
			\end{conditions}
		\end{theorem}
		For sufficiently expressive fragments, the operation of adding the local predicates corresponds to mapping the equivalent variety $\V$ to $\V*\D$.
		In fact, it will not be the case only if the fragment cannot use these predicates properly.
		In all cases, it is equivalent to adding the descriptive local predicates defined in Section~\ref{Section:Defs}.
		The proof of the following proposition follows the proof of Theorem~\ref{semimod}, by using an adapted
		notion of enriched alphabet. We omit the proof, that could be find in~\cite{PapermanPhd}.
	\begin{proposition}\label{Prop:localpredfragment}
			Let $\cF[\sigma]$ be a
			fragment of logic equivalent to a variety $\V$ and
			$L$ a regular language with $S$ as syntactic semigroup.
			The following conditions are equivalent.
			\begin{conditions}
				\item {$L$ is definable in $\cF[\sigma,+1]$.}
				\item $S$ belongs to $\V*\D$.
				\item $S_E$ belongs to $\gV$.
			\end{conditions}
		\end{proposition}
	\paragraph{The derived category relatively to modular languages.}
		Following the preceding paragraph, we give the definition of the \emph{derived category} adapted to
		modular languages which was largely inspired by the article of~\cite{CPS06b}.

			Let $\varphi: A^*\to M$ be a morphism and $d$ an integer.
			The \emph{$d$-derived category} of $\varphi$, denoted $C_d(\varphi)$, is
			the category with $\Zd$ as set of objects,
			and the arrows from $i$ to $j$ are the elements $m$ of $M$ such
			that there exists a word $u$ satisfying $\varphi(u)=m$ and $i+|u|\equiv j \bmod d$.
			The $d$-derived category of a regular language $L$, denoted $C_d(L)$, is the category $C_d(\eta_L)$.
			The following lemma is a straightforward consequence of the definition that will be of some use.
		\begin{lemma}\label{monLocaux}
		Let $d$ be a positive integer, and $L$ be a regular language of stability index $s$.
		Then the local monoids of $C_d(L)$ are isomorphic to $\eta_L((A^d)^*)$.
		In particular, the local monoids of $C_s(L)$ are isomorphic to the stable monoid of $L$.
		\end{lemma}

		\begin{example}\label{Excatder}
		The $4$-derived category of the language $(aa)^*ab(bb)^*$ is given below.
		Let $\eta$ be its syntactic morphism and $S$ its stable monoid. Its stability index is $4$.
		
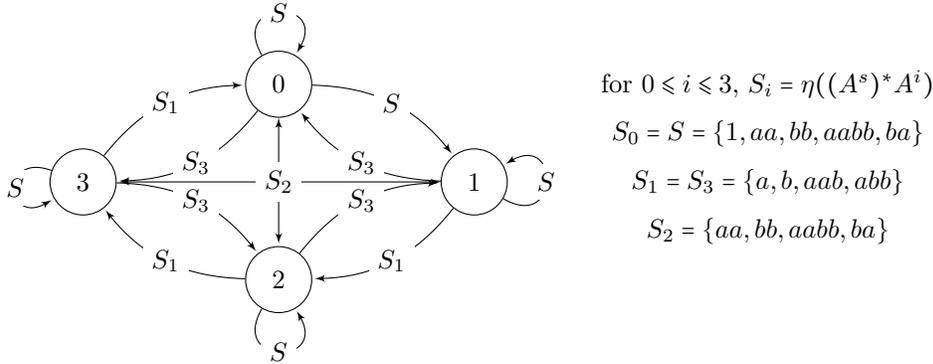
\begin{figure}[h]
\begin{tikzpicture}[scale=0.65,->,>=latex',shorten >=1pt,node distance=1.2cm, every loop/.style={looseness=5}]
\node [state](q_0){$0$};
\node [state](q_1) at ($(q_0) + (4,-2)$) {$1$};
\node [state](q_2) at ($(q_0) + (0,-4)$){$2$};
\node [state](q_3) at ($(q_0) + (-4,-2)$){$3$};
\node (S) at ($(q_0) + (10,-1)$) {$ S_0=S = \{1,aa,bb,aabb,ba\}$};
\node (SA) at ($(q_0) + (10,-2)$) {$ S_1=S_3 = \{a,b,aab,abb \}$};
\node (SA2) at ($(q_0) + (10,-3)$) {$ S_2 = \{ aa,bb,aabb,ba\}$};
\node (SA3) at ($(q_0) + (10,-0)$) {for $0\leq i\leq 3$, 
$S_i=\eta((A^s)^*A^i) $};
   \path[->]
   (q_0) edge [bend left = 25]  node[fill = white] {$S$} (q_1)
   (q_1) edge [bend left = 25] node[fill = white] {$S_3$} (q_0)
   (q_0) edge  [bend left = 25]node[fill = white] {$S_3$} (q_3)
   (q_3) edge [bend left = 25]  node[fill = white] {$S_1$} (q_0)
   (q_2) edge [bend left = 25]  node[fill = white] {$S_1$} (q_3)
   (q_3) edge [bend left = 25]  node[fill = white] {$S_3$} (q_2)
   (q_2) edge  [bend left = 25]	 node[fill = white] {$S_3$} (q_1)
   (q_1) edge [bend left = 25]  node[fill = white] {$S_1$} (q_2);
 
   \path[<->]
   (q_1) edge  node[pos = 0.5] {$ $} (q_3)
   (q_0) edge  node[pos = 0.5, fill = white] {$ S_2$} (q_2) ;

   \path[->]
   (q_0) edge  [in = 60, out = 120, loop] node[fill = white]  {$S$} ()
   (q_2) edge  [in = -60, out = -120, loop] node[fill = white]  {$S$} ()

   (q_1) edge  [in = 30, out = -30, loop] node[fill = white]  {$S$} ()
   (q_3) edge  [in = 210, out = 160, loop] node[fill = white]  {$S$} ();

\end{tikzpicture}

\end{figure}

		\end{example}
%
%

		\begin{proposition}\label{prop:division_cat}
		Let $L$ be a regular language.
		For any $0< d\leq d'$, if $d$ divide $d'$, then  $C_{d'}(L)$ divides $C_{d}(L)$.
		\end{proposition}
\begin{proof}
	Let $L$ be regular language
	and $0\leq d<d'$ be integers such that $d$ divides $d'$. We define the relation $\tau:C_{d'}(L)\to C_{d}(L)$.
			The object application $\Ob(\tau):\Z_{d'}\to \Z_{d}$ is defined by $\Ob(\tau)(x)=x\bmod{d}$
			for any $x\in \Z_{d'}$. Let $(x,m,y)$ be an arrow of $C_{d'}(L)$. By definition,
			there exists $u\in A^*$ such that
			$\eta_L(u) = m$ and $|u|\equiv y-x \bmod{d'}$.
			Let $a=x\bmod{d}$ and $b=y\bmod{d}$. Then, since $d$ divides $d'$,
			$|u|\equiv b-a\bmod{d}$. Thus,
			 the arrow
			$(a,m,b)$ is in $C_{d}(L)$.
			We define $\tau(x,m,y) = (a,m,b)$.
			The application $\tau$ is a morphism and for any
			$(x,m,y)\neq (x,m',y)$, we have
			$\tau(x,m,y)\neq \tau(x,m',y)$. Therefore,
			$\tau$ defines a division from $C_{d'}$ to $C_d$.
\end{proof}

		The derived category theorem was originally proved by~\cite{Tilson} for varieties of monoids and semigroups. Unfortunately the case of
		 modular languages can not be dealt with the framework of Tilson since they do not
		 form a variety of language. However it has been extended to \emph{length-multiplying} varieties in the PhD thesis of~\cite{ChaubPhD}. Since this work is only available in french, we provide a proof inspired by the work of Chaubard, but adapted to our framework.

		 \begin{theorem}	\label{thm:derivedcategory}
			Let $\cF[\sigma]$ be a fragment of logic equivalent to a variety of monoids $\V$,
			$L$ a regular language and $d$ a positive integer.
			Then the following properties are equivalent:
			\begin{conditions}
				\item\label{derivedcategory:1} $L$ is definable by a formula of $\cF[\sigma,\MOD^d]$,

				\item\label{derivedcategory:3} there exists
				      some languages $L_0,\ldots,L_{d-1}$ of
				      $\cF[\sigma]$ over $\Ae^*$ such that:
				      \begin{equation}\label{equationLd}
				      L =\bigcup_{i=0}^{d-1}\big( (A^d)^*A^i\cap \pi_d(L_i\cap \WF_d)\big)
				      	      \tag{a}
				      \end{equation}
				\item\label{derivedcategory:2} the category $C_d(L)$ belongs to $\gV$.
			\end{conditions}
		\end{theorem}
		\begin{proof}
			The equivalence between the two first points is obtained directly by Theorem~\ref{semimod}. We only
			prove the equivalence between~\eqref{derivedcategory:2} and~\eqref{derivedcategory:3}.
			As always, we denote by $\eta_L:A^*\to M_L$ the syntactic morphism of $L$ and $P = \eta_L(L)$ its
			accepting set.
%
%
%
				\begin{conditions}
				\item[$\eqref{derivedcategory:2}\to~\eqref{derivedcategory:3}$:]
				Assume that $C_d(L)$ belongs to $\gV$. By definition, it means that there exists a division of categories
				$\tau:C_d(L) \to M$, where $M$ is a monoid of $\V$ seen as a one object category.
				We need to define some appropriate languages $L_i$ for $0\leq i <d$.
				To this end,
				we construct an \emph{adequate} morphism from $\Ae^*$ to $M$. \\
				Let then $\beta: \Ae^* \to M$ be defined by $\beta(a,i) = m$
				where $m$ is any element in $\tau(i,\eta(a),i+1\bmod d)$.
			  For $0\leq i < d$, let $E_i = \bigcup_{m\in P}\tau(0,m,i)$ and
				$L_i=\beta^{-1}(E_i)$. Because
				$M$ is in $\V$, these languages are all in $\cF[\sigma]$.

				It remains to verify that these languages satisfy the hypothesis. This is equivalent to check that
				for all $i<d$
				$$\alpha_d(L\cap (A^d)^*A^i) \subseteq L_i\text{ and }\alpha_d(L^c\cap (A^d)^*A^i)\cap L_i =\emptyset.$$
				Let $u=(a_0,0)\cdots (a_n,p)$ be a well-formed word of $\Ae^*$, by construction of $\beta$,
				we have
				$$\beta(u) = m = m_1\cdots m_n \in  \tau(0,\eta_L(a_1),1)\cdots \tau(p,\eta_L(a_n),p+1) \subseteq \tau(0,\eta_L(a_1\cdots a_n),p+1)$$
				Therefore, we have
				$$\beta\Big(\alpha_d(L\cap (A^d)^*A^i)\Big) \subseteq E_i.$$
				Furthermore, since $\tau$ is a division, for all $u\in \alpha_d(L^c\cap (A^d)^*A^i)$,
				$\beta(u)\not \in \tau(0,m,i)$ for all $m\in P$ and thus $\beta(u)\not \in E_i$.

				\item[$\eqref{derivedcategory:3}\to~\eqref{derivedcategory:2}$:]
					Let $L_0,\ldots, L_{d-1}$ be languages of $\cF[\sigma]$ as stated by~\eqref{derivedcategory:3}.
					Then each of them is definable by a monoid of $\V$, and since varieties are closed by product, there exists a morphism $\beta:\Ae^*\to M$ that recognizes them all, with $M\in \V$.
					We now prove that $C_d(L)$ divides $M$.
					Let $\tau:C_d(L)\to M$  be defined by
					$$\tau(i,x,j) = \{ \beta(u)\mid \exists u\in \WF_d(i,j) \text{ s.t. } \eta_L(\pi_d(u)) = x\}$$
					The application $\tau$ satisfies the first three axioms of a division of categories.
					\begin{conditions}
						\item We have $1\in\tau(i,1,i)$ for any $i$ of $\Z/d\Z$.
						\item Let $(i,x,j)$ be an arrow of $C_d(L)$.
						By definition, there exists $v$ in $(A^d)^*A^{j-i}$ such that $\eta_L(v) = x$.
						Let $u=\alpha_d^i(v)\in \WF_d(i,j)$. By definition, $\beta(u)\in \tau(i,x,j)$ and thus
						$\tau(i,x,j)\neq \emptyset$.
						\item Let $(i,x,j)$ and $(j,x',k)$ be two arrows in $C_d(L)$ and
						$m\in\tau(i,x,j)$, $m'\in \tau(j,x',k)$. By hypothesis, there exists $u\in \WF_d(i,j)$ and
						$u'\in \WF_d(j,k)$ such that
						$\beta(u)=m$ and $\beta(u')= m'$, and such that $\eta_L(\pi_d(u)) = x$ and
						$\eta_L(\pi_d(u')) = x'$. Then, $mm'$ belongs to $\tau(i,xx',k)$ since
						$mm'= \beta(uu')$, $uu'\in \WF_d(i,k)$ and $\eta_L(\pi_d(uu')) = xx'$.
					\end{conditions}
					Unfortunately, it could happen that $\tau$ does not satisfy the last condition.
					Without detailing, this is due to the fact that
					some elements of the syntactic congruence of $L$ might merge when appearing at some specific congruences, leading to non empty intersection of images of arrows.
					In the following, we use the idea that for any pair of elements there exists a congruence that separates them by definition of the syntactic congruence.

					Thus, we now introduce a \emph{twisted product} of $\tau$, denoted by $\otimes_d\tau:C_d(L)\to M^d$
					and formally define it by
					$$\otimes_d\tau(i,x,j) = \big(\tau(i,x,j),\tau(i+1,x,j+1),\ldots,\tau(i+d-1,x,j+d-1)\big)$$
					Because $\otimes_d\tau$ is a product of $\tau$ by it self $d$ times, it satisfies immediately the
					first
					three axioms of a division of categories. We now prove that $\otimes_d\tau$ is a division by proving the separation axiom.
					\begin{conditions}
						\setcounter{enumi}{3}
						\item Let $x,x'$ be two distinct elements of $M_L$ such that
						$(i,x,j)$ and $(i,x',j)$ are arrows of $C_d(L)$. We first
						prove that there exists $r,t$ satisfying $r-t=j-i$
						and such that $\tau(t,x,r)\cap \tau(t,x',r) = \emptyset$
						and then conclude by using $\otimes_d\tau$.
						Let $v$ and $v'$ in $(A^d)^*A^{j-i}$ such that $\eta_L(v) = x$ and $\eta_L(v') = x'$.

						Since $x\neq x'$, and by definition of $M_L$, the syntactic monoid of $L$,
						we can assume that there exists $p,q\in A^*$ such that $pvq \in L$
						if and only if $pv'q\not\in L$.
						Let $y=\eta_L(pvq)$ and $y'=\eta_L(pv'q)$.
						We remark that $(0,y,k)$ and $(0,y',k)$ are
						arrows  $C_d(L)$ for $k=|pvq|\bmod{d}=|pv'q|\bmod{d}$.
						Without loss of generality, we assume $pvq$ to be in $L$, the other case being symmetrical.
						By hypothesis, we have the following:
						\begin{align*}
							\eta_L^{-1}(y)\cap& (A^d)^* A^k \subseteq L_k\\
							\eta_L^{-1}(y')\cap& (A^d)^* A^k \cap L_k = \emptyset
						\end{align*}
						However
						\begin{align*}
						\tau(0,y,k)&= \beta\circ\alpha_d\big(\eta_L^{-1}(y)\cap (A^d)^* A^k\big)\\
						\tau(0,y',k)&= \beta\circ\alpha_d\big(\eta_L^{-1}(y')\cap (A^d)^* A^k\big)
						\end{align*}

						Since $L_k$ is recognized by $M$ through the morphism $\beta$ we have
						$$\tau(0,y,k)\cap \tau(0,y,k)=\emptyset.$$
						To conclude, it suffices to notice that
						$$\tau(0,s,t)\cdot\Big(\tau(t,x,r)\cap \tau(t,x',r)\Big)\cdot\tau(r,t,k)\subseteq \tau(0,y,k)\cap \tau(0,y,k)= \emptyset,$$
						where $s=\eta_L(p)$, $t=|p|$, $r=t+j-i$ and $t=\eta_L(q)$. Since both $\tau(0,s,i)$ and
						$\tau(j,t,k)$ are nonempty, we conclude that $\tau(t,x,r)\cap \tau(t,x',r) = \emptyset$.
						We proved that for every arrow $(i,x,j)$ and $(i,x',j)$ in $C_d(L)$, there exists
						$r,t$ $r-t=j-i$ and such that $\tau(t,x,r)\cap \tau(t,x,r)=\emptyset$.
						Therefore, we obtain that  $\otimes_d\tau(i,x,j)\cap  \otimes_d\tau(i,x',j) =\emptyset$
						for every coterminal arrows $(i,x,j)$ and $(i,x',j)$ in $C_d(L)$, which concludes the proof.

					\end{conditions}

			\end{conditions}
		\end{proof}

\subsection{local case}\label{Subsection:Local}

	For any variety $\V$, we define $\QV$ to be the class of morphisms ($lm$-variety of morphisms to be precise,
	see the article of~\cite{PS05} for more details) whose stable
	monoid is in $\V$.
	Following the article of \cite{Tilson}, we denote by
	$\ell \V$ the variety of categories whose local monoids are all in $\V$.
	A variety of monoids $\V$ is said to be \emph{local} if $\gV = \lV$.
	The next theorem makes explicit the link between $\QV$ and $\ell\V$.

	\begin{theorem}\label{QVlV}
	Let $\V$ be a variety and $L$ a regular language of $A^*$ of
	stability index $s$. The following properties are equivalent:
	\begin{conditions}
	\item\label{qv} $L$ is recognized by a morphism in $\QV$,
	\item\label{dlv} there exists an integer $d$ such that $C_d(L)$ is in $\ell\V$,
	\item\label{slv} $C_s(L)$ is in $\ell\V$.
	\end{conditions}
	\end{theorem}
	\begin{proof}
	$\ $\\
	$\ref{qv}\to \ref{slv}$. If $L$ is recognized by a stamp in $\QV$, then its syntactic stamp is also in $\QV$
			and its stable monoid
			is in $\V$. But, thanks to Lemma \ref{monLocaux},
			the local monoids of $C_s(L)$  belong to $\V$, and thus $C_s(L)$ is in $\ell\V$.\\
	$\ref{slv}\to \ref{dlv}$. Is obvious.\\
	$\ref{dlv}\to \ref{qv}$. Suppose that $C_d(L)$ is in $\ell\V$.
			Then the local monoids of $C_d(L)$,
			which are isomorphic to $\eta_L((A^d)^*)$ by Lemma \ref{monLocaux}, belong to $\V$.
			Thus $\eta_L((A^{ds})^*)$, which is a submonoid of $\eta_L((A^d)^*)$, also
			belongs to $\V$. Finally, by definition of the stability index, the monoid
			$\eta_L((A^s)^*)=\eta_L((A^{ds})^*)$ is in $\V$ and thus
			$\eta_L$ is in $\QV$.
	\end{proof}

	Observe that any monoid of $\V$, viewed as a one-object category,
	belongs to $\ell\V$. Therefore by definition of $\gV$, any category of $\gV$
	divides a category of $\ell\V$, and thus $\gV\subseteq\ell\V$.
	The varieties satisfying $\gV=\ell\V$ are exactly the  local varieties.
	Combining this with Theorem~\ref{thm:derivedcategory} and
	 since the stability index and the stable monoid of
	 a given regular language are computable, one gets the following corollary.
	 \begin{corollary}
		Let $\cF[\sigma]$ be a fragment equivalent to a local variety $\V$. Then $\cF[\sigma]$ is
		decidable  if and only if $\cF[\sigma,\MOD]$ is decidable. Furthermore, the fragment $\cF[\sigma,\MOD]$ is
		equivalent to $\QV$.
	\end{corollary}

		Adding modular predicates does not always coincide with the $\mathbf{Q}$ operation.
		A counterexample is the variety $\J$, which is known to be nonlocal.
		\cite{CPS06b} proved the decidability of $\BS_1[<,\MOD]$, using the
		characterization of $g\J$ given by~\cite{Kna83} (see Figure~\ref{gJ}).
		Using this characterization, we can prove that the language $(aa)^*ab(bb)^*$,
		whose stable monoid is in $\J$ does not satisfy Knast's equation since
		\begin{align*}
		(m_1m_2)^\omega (m_3m_4)^\omega&=(aa)^\omega(bb)^\omega=aabb\\
		&\neq ab =(aa)^\omega ab(bb)^\omega= (m_1m_2)^\omega m_1m_4(m_3m_4)^\omega
		\end{align*}
		It is therefore not definable in $\BS_1[<,\MOD]$ (see Example~\ref{Excatder}).
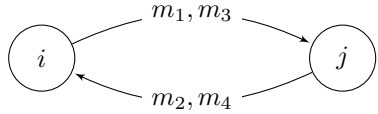
\begin{figure}[H]
\centering
\begin{tikzpicture}[->,>=latex',shorten >=1pt,node distance=1.2cm,  every loop/.style={looseness=5}]
\node [state](q_0){$i$};
\node [state](q_1) at ($(q_0) + (4,0)$) {$j$};
   \path[->]
   (q_0) edge [bend left = 25]  node[fill = white] {$m_1,m_3$} (q_1)
   (q_1) edge [bend left = 25] node[fill = white] {$m_2,m_4$} (q_0);

\node (equa) at ($(q_0) + (2,-2)$) {$(m_1m_2)^\omega (m_3m_4)^\omega=(m_1m_2)^\omega m_1m_4(m_3m_4)^\omega$};
\end{tikzpicture}
\caption{Path equation of $\mathbf{g}\J$ by Knast. }\label{gJ}
\end{figure}

\subsection{Finite rank}\label{Subsection:FiniteRank}

Although the local property gives a nice algebraic characterisation, it only applies to a few varieties.
Nonetheless, we can still obtain a delay when the global is well-understood.
To be more precise, we now prove a delay for varieties where equations for the global are known.
As the global is a variety of categories, we first extend the framework of profinite equations to categories.
Note finally that this is the only case where we obtain a delay that is greater than the stability index.
The main applications on fragments of logic are given in Corollary~\ref{FinRank-CorDecid}.

\paragraph{Path equations}
		The theory of profinite equation of varieties of monoids extends naturally to path equations on graphs,
		characterising varieties of categories.
		The complexity of a variety of categories is given by its rank,
		which is the minimal size required to describe the variety in terms of path equations.
			Let $X$ be a graph and $E$ the set of arrows of $X$. Then $X^*$
			is the set of words on $u=u_0\cdots u_n\in E^*$ such that for all
			$i<n$, $u_i$ and $u_{i+1}$ are consecutive arrows.
			$X^*$ is named the free category on $X$.
			Let $u$ and $v$ be coterminal paths of $X^*$. Then
			$$r(u,v)=\min{\big\{n\mid \exists \varphi:X^*\to C\text{ with }n=|C|\text{ and }\varphi(u)\neq \varphi(v)\big\}}$$
			where $\varphi$ is a category morphism and $C$ a finite category.
			We define $d(u,v)=2^{-r(u,v)}$ which is an ultrametric distance  on $X^*$.
			The completion of $X$ for this metric is called the profinite free category on $X$ and is denoted by $\widehat{X}^*$.
			The following proposition is very standard in the framework of (pseudo-)varieties of monoids and categories.
		\begin{proposition}
			Let $X$ be a graph, $C$ a finite category and $\varphi:X^*\to C$ a morphism of categories.
			Then, there exists a unique continuous function $\widehat{\varphi}:\widehat{X^*} \to C$ that extends $\varphi$.
			Furthermore, for any $u\in \widehat{X^*}$, there exists $v\in X^*$ such that $\widehat\varphi(u) = \widehat\varphi(v)$.
		\end{proposition}

			Let $X$ be a graph and $u,v\in \widehat{X^*}$ coterminal profinite paths.
			We say that the finite category $C$ satisfy the equation $(X,u=v)$ if for any morphism
			$\varphi: X^*\to C$, we have $\widehat{\varphi}(u)=\widehat{\varphi}(v)$.
		\begin{theorem}[Tilson]
			Every non trivial variety of finite categories is defined by a set of equations. 
		\end{theorem}
		\begin{definition}[Rank of a variety]\label{def:rank}
			We say that a variety of monoids $\V$ has a rank $k$  if its global is defined by a set of bounded path
			equations with at most $k$ vertices. If $\V$ has a finite rank, we denote by
			$\rank(\V)$ the minimal $k$ such that $\V$ has a rank $k$.
		\end{definition}
		We remark that the varieties of rank one are exactly the local ones.
		Furthermore, most of the known fragments of logic are equivalent to a variety of finite rank. The question remains however
		open in some cases, as for instance for the levels of the dot-depth hierarchy.
		\begin{example}\label{FinRank-Ex}
			We now  give several varieties where equations for the global are known.
			\begin{enumerate}
				\item Several varieties are known to be local. For instance, the variety of semilattice monoids $\J^1=\equations{xy=yx, x^2=x}$,
				the variety $\DA=\equations{(xy)^\omega x (xy)^\omega = (xy)^\omega}$,
				the variety of aperiodic monoids $\A = \equations{x^\omega = x^{\omega+1}}$.
				\item The variety of commutative monoids  $\Com=\equations{xy=yx}$. The variety of categories
				$g\Com$ is defined below and thus is of rank $2$.	

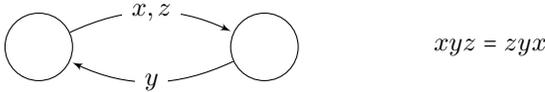
\begin{tikzpicture}[->,>=latex',shorten >=1pt,node distance=0.5cm, every loop/.style={looseness=5}]
\node [state](q_0){};
\node [state](q_1) at ($(q_0) + (3,0)$) {};
\node  at ($(q_1) +(3,0)$) { $xyz = zyx$};
   \path[->]
   (q_0) edge [bend left = 25]  node[fill = white] {$x,z$} (q_1)
   (q_1) edge [bend left = 25] node[fill = white] {$y$} (q_0);  
\end{tikzpicture}
%
%
%
%
%
%
%
%

				\item A recent algebraic description of
				the languages definable by formulas of $\BS_{k+1}^2[<]$
				was established in~\cite{KS12,KW12}.
				 In the subsequent we will denote by $\V_k$ the equivalent variety of monoids.
				This result was extended to $\BS_{k+1}^2[<,\loc]$ in~\cite{KL12}. From this latter result we derive
				the following description of $\gV_k$, giving a rank of at most $2k$.
			\end{enumerate}
							\begin{figure}[h!]
\begin{tikzpicture}[->,>=latex',shorten >=1pt,node distance=1.2cm, every loop/.style={looseness=5}]
\node [state](q_0){};
\node [state](q_1) at ($(q_0) + (4,0)$) {};
\node  at ($(q_1) +(5,5)$) {We define $U_1 = (sx_1)^\omega s(y_1t)^\omega$ and $V_1=(sx_1)^\omega t(y_1t)^\omega$};
\node  at ($(q_1) +(5,4)$) {$U_k = (p_kU_{k-1}q_kx_k)^\omega p_k U_{k-1}q_k(y_{k}p_kU_{k-1}q_k)^\omega$} ;
\node  at ($(q_1) +(5,3)$) {$V_k = (p_kU_{k-1}q_kx_k)^\omega p_k V_{k-1}q_k(y_{k}p_kU_{k-1}q_k)^\omega$} ;
\node  at ($(q_1) +(5,0)$) {$g\V_k$ satisfies the equation ${U_k = V_k}$} ;

\node  [state](q_2) at ($(q_0) +(0,2)$) { };
\node  [state](q_3) at ($(q_1) +(0,2)$) { };
\node  [state](q_4) at ($(q_2) +(0,3)$) { };
\node  [state](q_5) at ($(q_3) +(0,3)$) { };
\node (p2) at ($(q_2)+(0,1)$){$p_3$}; 
\node (p3) at ($(q_3)+(0,1)$){$q_3$}; 
\node (pm) at ($(q_4)-(0,1)$){$p_k$}; 
\node (qm) at ($(q_5)-(0,1)$){$q_k$};
\node (p4) at ($(q_2)+(2,2)$){$\vdots$}; 
\node (p22) at ($(p2)+(0,0.5)$){$\vdots$}; 
\node (p33) at ($(p3)+(0,0.5)$){$\vdots$}; 
%

   \path[->]
   (q_0) edge [bend left = -25]  node[fill = white] {$s,t$} (q_1)
   (q_1) edge [bend left = -25] node[fill = white] {$x_1,y_1$} (q_0)
   (q_1) edge[bend left = -25]  node[fill = white] {$q_2$} (q_3)
   (q_3) edge [bend left = -25] node[fill = white] {$x_2,y_2$} (q_2)
   (q_5) edge [bend left = -25] node[fill = white] {$x_k,y_k$} (q_4)
   (q_2) edge [bend left =- 25] node[fill = white] {$p_{2}$} (q_0)
   (q_3) edge  (p3)
   (p2) edge  (q_2)
   (q_4) edge  (pm)
   (qm) edge  (q_5)
 ;


;

\end{tikzpicture}
\caption{An equation for $\gV_k$.}\label{Fig:gVk}
\end{figure}

		\end{example}

	\begin{theorem}[A  Delay Theorem for finite rank varieties]\label{delay-FinRank}
			Let $\cF[\sigma]$ be a fragment equivalent to a variety $\V$ rank $k$.
			A language $L$ belongs to $\cF[\sigma,\MOD]$ if and only if $L$ belongs to
			$\cF[\sigma,\MOD^{ks}].$
		\end{theorem}

		\begin{proof}
	First notice that since the if condition is trivial, we only need to prove the only if implication.
	Remark now that if $\rank(\V)=1$ then the variety is local and we know that we can restrict to congruence modulo the stability index.
	For the rest of the proof we assume that $\rank(\V) = k>1$.\\
	Let now $d$ be such that $C_d(L) \in \gV$.
	Without loss of generality, we assume that $d$ is greater than $k$.
	Indeed if $d\leq k$, we consider $d'=dk$. Then by Proposition~\ref{prop:division_cat} $C_{d'}(L)$ divides $C_d(L)$ and thus also belongs to $\gV$. \\
	 So in the remainder of the proof we will assume that $ds>{ks}$.
	 Since $C_d(L)\in \gV$ we know that $C_{ds}(L)\in \gV$.
		Then $C_{ds}(L)$ satisfies every path equation $(X,u=v)$ defining $\gV$.
	The goal of this proof is to show if $C_{ks}(L)$ does not satisfy a path equation defining $\gV$, then $C_{ds}(L)$ cannot satisfy it either.

		So assume that there exists a path equation $(X,u=v)$ of rank $k$ defining $\gV$ that is not satisfied by $C_{ks}(L)$.
		Then, there exists a category morphism $\varphi:X^*\to C_{ks}(L)$
		such that $\widehat{\varphi}(u)\neq \widehat{\varphi}(v)$.
		We define $V=\varphi(\Ob(X))$ the set of objects of $C_{ks}(L)$ that have a preimage by $\varphi$, and
		$$E=\{(i,m,j)\in C_{ks}(L) \mid \exists e\in X \ \varphi(e)=(i,m,j)\}$$
		 the set of arrows that have a preimage by $\varphi$.
		Notice that $E\subseteq V\times M_L\times V$.

		We will construct a category morphism $\psi:X^*\to C_{ds}(L)$
		such that $\widehat{\psi}(u)\neq \widehat{\psi}(v)$.
		In order to do that, we define a map $\theta:V\to C_{ds}(L)$
		such that for all $(i,m,j)$ in $E$, $(\theta(i),m,\theta(j))$ is an arrow of $C_{ds}(L)$.

		\begin{lemma}
		There exists a smallest integer $i_V<{ks}$ such that
		 $\{i_V+1,\ldots,i_V+s-1\bmod{ks}\}\cap V=\emptyset$.
		\end{lemma}
		\begin{proof}
		As the size of $X$ is $k$, the size of $V$ is at most $k$.
		Then the maximal distance between two consecutive vertices of $V$
		is at least $ks/k=s$.
		\end{proof}

		We define $\theta:V\to\Ob(C_{ds}(L))$ as follow :
		$$\theta: \begin{cases}
			i \mapsto i \bmod{ds} \text{ if }i\leq i_V \\
			i \mapsto ds+i-{ks}  \text{ otherwise}.
		\end{cases}$$
		The idea behind this is that $i=\theta(i)$ if $i$ appears before the gap and ${ks}-i= ds-\theta(i)$ if $i$ appears after it.
		Then each arrow from $E$ will either appear directly as it does for $C_{ks}(L)$ if it does not go over the gap, and since the gap is of size $s$, we will be able to pump the arrows that go over it.

		\begin{lemma}\label{lemma:theta_arrow}
		For any arrow $(i,m,j)$ of $E$, $(\theta(i),m,\theta(j))$ is an arrow of $C_{ds}(L)$.
		\end{lemma}
		\begin{proof}
		Let $(i,m,j)$ be an arrow of $E$. Then there exists a word $u$ such that
		$\eta_L(u)=m$ and $i+|u|=j \bmod{ks}$.
		We now distinguish the cases depending on the length of $u$.
		\begin{itemize}
		\item If $|u|\geq s$, then we know, by definition of the stability index, that for any positive integer $\ell$, there exists a word $u_\ell$ such that
		$\ell s\leq |u_\ell| <(\ell+1)s$, $|u|=|u_\ell|\bmod s$ and $u\equiv_L u_\ell$.
		Then as $\theta$ preserves the congruence modulo $s$,
		$(\theta(i),\eta_L(u_\ell),\theta(j))=(\theta(i),m,\theta(j))$
		is an arrow of $C_{ds}(L)$.
		\item If $|u|<s$, then we have to treat several subcases:
			\begin{itemize}
				\item If $\theta(i)=i$ and $\theta(j)=j$, then $\theta(i)+|u|=\theta(j)\bmod{ds}$.
				Thus $(\theta(i),m,\theta(j))$ is an arrow of $C_{ds}(L)$.
				\item If $\theta(i)=ds+i-{ks}$ and $\theta(j)=ds+j-{ks}$,
		 then as $u$ has a size smaller than $s$, we have $i<j$ and $\theta(j)-\theta(i)=j-i$. Consequently $\theta(i)+|u|=\theta(j)\bmod{ds}$ and
		  		$(\theta(i),m,\theta(j))$ is an arrow of $C_{ds}(L)$.
		  		\item If $\theta(i)=ds+i-{ks}$ and $\theta(j)=j$, then
		  		$i+|u|=j+{ks}$.
		  		So  $\theta(i)+|u|= ds+i-{ks} +|u|=j+ds$.
		  		The same word $u$ labels an arrow from $\theta(i)$ to $j$ and thus
		  		$(\theta(i),m,\theta(j))$ is an arrow of $C_{ds}(L)$.
		  		\item Finally, the case where  $\theta(i)=i$ and $\theta(j)=ds+j-{ks}$ cannot happen since
		  		it implies that $i\leq i_V$ and $j>i_V+s$, and that
		  		$|u|=j-i> s \bmod{ks}$ which contradicts the $|u|<s$ hypothesis.

			\end{itemize}
		\end{itemize}
		\end{proof}

		We now define a new morphism $\psi:X^*\to C_{ds}$. We proceed as follow:
		\begin{itemize2}
				\item First we define $\Ob(\psi)$ to be $\theta\circ \Ob(\varphi)$.
				\item We now have to define $\psi$ on arrows.
				Let $e$ be an arrow of $X$ and $\varphi(e) = (i,m,j)$.
				We set $\psi(x) = \big(\theta(i),m,\theta(j)\big)$. This is well defined thanks
				to Lemma~\ref{lemma:theta_arrow}.
		\end{itemize2}
		\begin{lemma}\label{lemma:psi}
				Let $u$ be a path in $X^*$. If $\varphi(u) = (i,m,j)$, then $\psi(u) = \big(\theta(i),m,\theta(j)\big)$.
		\end{lemma}
		\begin{proof}
			Let $u=u_1\cdots u_n\in X^*$ such that $\varphi(u_\ell)=(i_\ell,m_\ell,j_\ell)$ and $\varphi(u) = (i,m,j)$. Therefore,
			$\psi(u_\ell)=\big(\theta(i_\ell),m_\ell,\theta(j_\ell)\big)$. However, since for all $1\leq \ell<n$
			$j_\ell = i_{\ell+1}$, we have
			$\phi(u) = (i_1,m_1\cdots m_n,j_n)=(i,m,j)$ and
			$\psi(u) = \big(\theta(i_1),m_1\cdots m_n,\theta(j_n)\big)=\big(\theta(i),m,\theta(j)\big).$
		\end{proof}

		Recall that $\widehat{\varphi}(u) \neq \widehat{\varphi}(v)$.
		Then we can find $u',v'\in X^*$ co-terminal paths of $X^*$ such that $\varphi(u')=\widehat{\varphi}(u)$,
		 $\varphi(v')=\widehat{\varphi}(v)$, $\psi(u')=\widehat{\psi}(u)$
		and $\psi(v')=\widehat{\psi}(v)$.
		 We set $u'=u_1\cdots u_n$ with $u_i\in X$ for any $i$ and
		$v'=v_1\cdots v_p$ with $v_i\in X$ for any $i$.
		To conclude we show that $\psi(u') \neq \psi(v')$ which is absurd. Indeed, if $\varphi(u')=(i,m,j)\in C_{{ks}}(L)$
		and $\varphi(v')=(i,m',j)\in C_{{ks}}(L)$, then $m'\neq m$ since $\widehat{\varphi}$ separates $u$ and $v$. Furthermore, by Lemma~\ref{lemma:psi},
		 we also have $\psi(u')=\big(\theta(i),m,\theta(j)\big)$ and $\psi(v')=\big(\theta(i),m',\theta(j)\big)$ in $C_{sd}(L)$.
		 Finally $\psi(u')\neq \psi(v')$ and thus $C_{ds}$ does not satisfy $(X,u=v)$, holding a contradiction.

		\end{proof}

		Combining the previous theorem with the decidable path equations given in Example~\ref{FinRank-Ex} yields the following corollaries.

		\begin{corollary}\label{FinRank-CorDecid}
		Given a regular language $L$ of stability index $s$ and an integer $k>0$. We have the following results.
		\begin{itemize}
			\item  $L$ belongs to $\FO^2_k[<,\MOD]$ if, and only if, it belongs to
		$\FO^2_k[<,\MOD^{2ks}]$.
			\item $L$ belongs to $\FO[=,\MOD]$ if, and only if, it belongs to	 $\FO[=,\MOD^{2s}]$.
		\end{itemize}

		As the corresponding global varieties are decidable, we get that the fragments $\FO[=,\MOD]$ and $\FO^2_k[<,\MOD]$ for any $k>0$ are decidable.
		\end{corollary}

\subsection{Infinitely testable case}\label{SsSection:InfTest}
	In this Section, we present the \emph{infinitely testable} property.
	We then prove that for any expressive enough fragment equipped with all regular predicates, this property holds, leading to a delay.
	In fact, Proposition~\ref{Prop:VD} proves that, given Proposition~\ref{Prop:localpredfragment}, as soon as a fragment contains the local predicates, it will be infinitely testable.
	Theorem~\ref{thm:delayinftest} then proves that a delay can be computed in this latter case.
	Informally, a variety is infinitely testable if the membership of a language
	to the variety only depends on words \emph{long enough}.
	\paragraph{Definition.}

	Given a semigroup $S$, the \emph{idempotents' ideal} of $S$, denoted $\tinf S$,
	is the ideal of $S$ generated by its idempotents, i.e. $\tinf S=SE(S)S$, where $E(S)$ denotes the set of idempotents of $S$.
	Note also that given a morphism
	 $\eta:A^+\to S$, it is the semigroup
	of all elements of $S$ having an infinite number of preimages by $\eta$.
	An aware reader could notice that $\tinf S$ is the set of all elements of $S$ that
	are $\mathcal J$-below an idempotent.
	A variety of semigroups $\V$ is said to be \emph{infinitely testable}
	if the membership of a semigroup to $\V$ is equivalent to
	the membership of its idempotents' ideal.
	Informally, a variety is infinitely testable if its membership
	can be reduced to an algebraic condition on the idempotents' ideal.
	By extension, we say that a fragment of logic is infinitely testable if
	it is characterized by an infinitely testable variety.
	\begin{example}
	The fragment $\FO[=]$ is equivalent to the
	aperiodic and commutative variety $\ACOM$. This fragment is also described by the equations
	$xy=yx$ and $x^{\omega+1} = x^\omega$. This fragment is not infinitely
	testable. For instance the language equal to the singleton $\{ab\}$ has a trivial idempotents' ideal while it is
	not
	definable in $\FO[=]$.
	\end{example}

\begin{example}\label{Prop:FO1-InfTest}
	The fragment $\FO[+1]$ is equivalent to the languages whose syntactic semigroup belongs to
	the variety: $\ACOM*\D$~\citep[Theorem VI.3.1]{Straubing94}.
	This fragment is also described by the profinite equation
	\begin{align*}
	x^\omega u y^\omega v x^\omega w y^\omega =x^\omega w y^\omega v x^\omega u y^\omega\enspace.  \tag{a}\label{acomli}
	\end{align*}
	We now show that it is an infinitely testable fragment.
	Let $L$ be a regular language and $S$ its syntactic semigroup.
	We simply prove that if the equation~\eqref{acomli} is not satisfied by $S$, then it is not satisfied by
	$\tinf{S}$.
	Suppose that there exists $x,y,u,v,w\in S$ such that the  equation~\eqref{acomli} is not satisfied.
	Then by setting:
		$x' = x^\omega,\ 	y' = y^\omega,\
		u' = x^\omega u y^\omega, \  v' = y^\omega v x^\omega, \   w' = x^\omega w y^\omega.$
	All new variables belong to $\tinf S$ and they also fail to satisfy~\eqref{acomli}.
\end{example}

	In fact, the approach given in the last example can be generalised to any variety of the form $\V*\D$.
	This is proved by Proposition~\ref{Prop:VD}.

%
%
%
%
%

\begin{proposition}\label{Prop:VD}
	 Let $\V$ be a variety.
		The variety $\V*\D$ is infinitely testable.
	\end{proposition}
	\begin{proof}
		Let $L$ be a regular language with $\eta:A^*\to M_L$ its syntactic morphism and $S=\eta_L(A^+)$
		its syntactic semigroup.
		Using Theorem~\ref{thm:delaydefinite}, we have that $S$ belongs to $\V*\D$ if and only if
		$S_E$ belongs to $\gV$. To conclude, we just notice that by definition, $(\tinf{S})_E = S_E$, and therefore
		$\V*\D$ is infinitely testable.
	\end{proof}

We finally prove here that if a fragment is equivalent to a variety whose global is infinitely testable, then we can effectively compute a delay, which furthermore is independent from the fragment.
For varieties of the form $\V*\D$, this also gives the decidability thanks to Proposition~\ref{Prop:VD} and
Corollary~\ref{Cor:transfer}.

\begin{theorem}\label{thm:delayinftest}
	Let $\cF[\sigma]$ be a fragment equivalent to a variety $\V$ which is not a group variety and
	let $L$ be a language of stability index $s$. Then $L$ belongs to $\cF[\sigma,+1,\MOD]$ if and only if
	$L$ belongs to $\cF[\sigma,+1,\MOD^s]$.
\end{theorem}
\begin{proof}
	First by Theorem~\ref{semimod}, a language $L$ belongs to $\cF[\sigma,+1,\MOD]$ if and only if there exists
	$d>0$, $L_0,\ldots,L_{d-1}$ in $\cF[\sigma,+1]$ such that
	$$L = \bigcup_{i<d} (A^d)^*A^i \cap \pi_d(L_i\cap \WF_d).$$
	Because $\cF[\sigma]$ is a fragment which is a variety of monoids but not a group variety, the
	language $\WF_d$ and $\textbf{max}$ belongs to $\cF[\sigma,+1]$.
	 We recall that $L_d=\pi_d^{-1}(L)\cap \WFp{d}$ for any $d>0$.
	Thus, for $i<d$,  $L_i\cap \WF_d$ belongs to $\cF[\sigma,+1]$ and we have the equality
	$$L_ = \bigcup L_i \cap \Ae(0,i) \cap \WF_d\in \cF[\sigma,+1].$$
	Therefore, $L$ belongs to $\cF[\sigma,+1,\MOD]$ if and only if there exists $d$ such that
	$L_d$ belongs to $\cF[\sigma,+1]$.
		Thus, it suffices to prove that
		if $L_{ds}$ is definable in $\cF[\sigma,+1]$, then
		$L_s$ is in $\cF[\sigma,+1]$ as well.
		We set $\eta_s:A_s^+\to S_s$ and $\eta_{ds}\colon A_{ds}^+\to S_{ds}$  the syntactic morphisms
		of $L_s$ and $L_{ds}$ respectively.

		\noindent\textbf{Claim.} The semigroup $\tinf{S_s}$ divides $\tinf{S_{ds}}$.\\

		Before proving this claim, let us remark that since a variety of semigroups is closed by division,
		this claim ends the proof. Since
		if $L$ belongs to $\cF[\sigma,\MOD^{ds}]$ then $S_{ds}$ belongs to $\V*\D$ and therefore
		$\tinf{S_{ds}}$ belongs to $\V*\D$ as well.
		By division, $\tinf{S_s}$ belongs to $\V*\D$, and thanks to Proposition~\ref{Prop:VD}, $S_s$ belongs to $\V*\D$. Finally, we deduce that $L_s$
		belongs to $\cF[\sigma+1]$.

		We now aim to construct a division from $\tinf{S_s}$ to $\tinf{S_{ds}}$.
		This is done through the enriched alphabet.
		We introduce the following projection
		$$h\colon\left\{\begin{array}{ccc}
			A_{ds}^+&\to& A_{s}^+\\
			(a,i)&\mapsto& (a,i\bmod{s})
			\end{array}\right.$$
		and $\WFpp{d}$ the language of \emph{well-formed factors},
		 which is the set of well-formed words that do not necessarily start by a letter $(a,0)$.
		 Note that
		${L_{ds} = h^{-1}(L_s)\cap \WFp{s}}$.
		Let us remark also that the image of a word not in $\WFpp{s}$ (resp. $\WFpp{ds}$)  by $\eta_s$ (resp. $\eta_{ds}$)
		has
		an absorbing zero as image by $\eta_s$ (resp. $\eta_{ds}$).
		This zero being idempotent, it belongs to $\tinf{S_s}$ (resp. $\tinf{S_{ds}}$).
		Finally, if two words of $\WFpp{s}$ have the same image by $\eta_s$, then
		they have the same length modulo $s$ and their first (and consequently last) letters have the same enrichment.

		Consider then $x$ a non-zero element of $\tinf{S_s}$.
		We show that
		$$h\inv(\eta_s\inv(x))\cap \eta\inv_{ds}(\tinf{S_{ds}})\neq \emptyset\enspace.$$
		Since $x$ belongs to $\tinf {S_s}$, there exists a word $u$
		of $A_s^+$ of length greater than $s$ in the preimage of $x$.
		And since $\eta_s(A_s^s)=\eta_s(A_s^{2s})$ by definition of the stability index,
		for any $k>0$ there exists a word $v_k$ of $A_s^+$ of length greater than $ks$ such that
		$u\equiv_L v_k$ and $|u|=|v_k| \bmod s$, since
		$\eta_s(u)=\eta_s(v_k)$.
		Then for $k$ sufficiently large, there exists a word $w$ in $h\inv(v_{k})$,
		such that $\eta_{ds}(w)$ belongs to $\tinf{S_{ds}}$.
		Note that by taking $k$ as a multiple of $d$, we obtain a word $w$ such that
		$|u|\bmod s= |w|\bmod ds$.
		Thus for each element $x\in \tinf{S_s}$, we can choose such an element, that we denote $w_x$.
		This justifies the definition of the following function:
		$$f\colon\left\{\begin{array}{cccl}
		 \tinf{S_s}&\to& \tinf{S_{ds}}\\
		 x &\mapsto& \eta_{ds}(w_x)&\text{ if }x\neq0\\
		 0&\mapsto & 0&\text{ otherwise.}
		\end{array}\right.$$
		We conclude by proving that $f$ is an injective morphism, and thus
		$\tinf{S_s}$ is a subsemigroup of $\tinf{S_{ds}}$.
	\begin{bulitem}
		\item[The application $f$ is a morphism.] Let $x,y\in \tinf{S_s}$. We show that
		$f(xy)=f(x)f(y)$.
		First, we can assume without loss of generality that $x\neq 0$ and $y\neq 0$.
		We remark that since $|w_x|\bmod ds=|h(w_x)|\bmod s$,
		the concatenated word $w_xw_y$ is well-formed if, and only if,
		$h(w_x)h(w_y)$ is well-formed too.
		If $xy\neq 0$.Then, $xy$ have
		a well-formed preimage
		and $w_xw_y$ is well-formed.
		Then as $w_{xy}$ and $w_xw_y$ are syntactically equivalent with respect to
		both $\WFpp{ds}$ and $h^{-1}(L_s)$,
		$\eta_{ds}(w_{xy}) =\eta_{ds}(w_xw_y)=\eta_{ds}(w_x)\eta_{ds}(w_y),$ meaning that $f(xy)=f(x)f(y)$.

		Now if $xy=0$, then either $xy$ has no well-formed preimage or
		$xy$ is a zero for $\pi_s\inv(L)$.
		In the latter case, then $f(x)f(y)=0$ according to the previous point.
		If $xy$ has no well-formed preimage, then $w_xw_y$ is not well-formed
		and consequently $f(x)f(y)=0$.
		\item[The application $f$ is injective.]
		Let $x,y\in \tinf{S_s}$ be such that $x\neq y$. Without loss of generality, we  assume that $x\neq 0$.
		Necessarily, there exist $p,q\in S_s$ such that $pxq\in \eta_s(L_s)$ if, and only if, $pyq\not\in \eta_s(L_s)$.
		Let $u$ and $v$ be words from the preimage of $p$ and $q$ respectively.
		Then there exists two words $u'\in h^{-1}(u)\cap \WFpp{ds}$ and $v'\in h^{-1}(v)\cap\WFpp{ds}$ such that
		$u'w_xv'\in L_{ds}$ if, and only if, $u'w_yv'\not\in L_{ds}$. Therefore, we have $f(x)\neq f(y)$ and $f$ is injective.
	\end{bulitem}
	\end{proof}
\\

The following proposition deals with fragments which are not varieties of groups. Varieties of groups are notoriously ill behaving with respect to their global. Indeed ~\cite{Auinger10} exhibited a variety of group $\mathbf{H}$ such that
$\mathbf{g(LH)}$ is undecidable (as a variety of \emph{semigroupoids}). However, for a local variety $\V$ which is not a variety of groups, the
variety of semigroups $\LV$ is local, as proved in~\cite{PapermanPhd}. Since this article does not deal with the framework
of varieties of \emph{semigroupoids}, we provide a self contain proof extracted from this latter result.

\begin{proposition}\label{prop:QLV}
		Let $\cF[\sigma]$ be a
		 fragment equivalent to a local variety $\V$ which is not a variety of groups.
		Then $\cF[\sigma,+1,\MOD]$
		is equivalent to $\QLV$.
\end{proposition}
\begin{proof}
First we remark that since $\cF[\sigma]$ is equivalent to a local variety, by
Proposition~\ref{Prop:localpredfragment}, and by definition of locality,
 $\cF[<,+1]$ is equivalent to $\LV=\V*\D$. Furthermore, since $\V$ is not a variety of groups,
 $(ab)^*$ belongs to $\LV$. Therefore, we obtain the following:
 $$L\in\cF[<,+1,\MOD]\underbrace{\text{ if and only if }}_{\text{By Theorem~\ref{thm:delayinftest}}}
 L \in \cF[<,+1,\MOD^s]\underbrace{\text{ if and only if }}_{\text{By Theorem~\ref{semimod}}} L_s \in \cF[<,+1]$$

	\noindent{\textbf{Claim:}} $L_s$ belongs to $\LV$ if and only if
	$L$ belongs to $\QLV$.

	We now prove both implications of the claim. In the following $S_s$ will be the syntactic semigroup
	of $L_s$ and $S$ the one of $L$.
	\begin{itemize}
	\item Assume that $L_s$ belongs to $\LV$.
	 Let $T= (A_s^s)^+\cap \WFp{s}$. We remark that $T$
	is a semigroup. Therefore, the set $\eta_s(T)$ is a subsemigroup of $S_s$. Since $S_s$ belongs to $\LV$,
	the semigroup $\eta_s(T)$ belongs to $\LV$ as well.
	 Remark now that $S_s$ is a quotient of the product of $S$ and the syntactic semigroup of
	$\WFp{s}$. Since the image of $\pi_s(T)$ in the syntactic monoid of $L$ is the stable semigroup of $L$ and
	 the image of $T$ in the syntactic
	semigroup of $\WFp{s}$ is trivial, we can conclude as $\eta_s(T)$ is isomorphic to the
	stable semigroup of $L$.

		\item Assume that $L$ belongs to $\QLV$, and we denote by $T$ its stable semigroup.
		By hypothesis,
		 $T$ is in $\LV$. One can remark that since $\V$ is not a variety of groups, it contains the semigroup
		 $U_1=\{0,1\}$ (equipped with the integer multiplication).
		 Therefore, the semigroup $T\cup\{0\}$, obtained by adding an
		absorbing element, also belongs to $\LV$. Indeed, it divides $T\times U_1$.

		 We now have to show that $L_s$ is in $\LV$ as well.
		Let $e$ be an idempotent of $S_s$. First, if $e$ is the zero of $S_s$, then
		$eS_se=\{e\}$. Otherwise, $e$ is the image of a well-formed factor $u$ that starts by a letter of
		the form $(a,i)$ and ends by a letter of the form $(a,j)$ with $j+1\equiv i\bmod{s}$.
		We denote by		$f$ the image of $\pi_s(u)$ by the syntactic morphism of $L$. This element
		is idempotent and, therefore, belongs to $T$.
		We conclude by noting that
		the local monoid  $eS_se$ is a quotient of $fTf\cup\{0\}$.
	\end{itemize}
\end{proof}

\section{Conclusion}\label{Section:Concl}
In this paper, we studied the definability problem for fragments of logic enriched with the modular predicates.
We presented a generic approach that gives the decidability of this problem in many cases, while the main applications are to the alternation hierarchies of the first order logic and its two variables counterpart.

The global approach is divided in two steps.
The first one relies entirely on logic.
We prove that adding a finite set of modular predicates preserves the decidability, given that the fragment is expressive enough.

The second part, which we call the delay problem, consists in deciding which finite set of modular predicates should be added to express a given regular language.
This is the most intricate part of the paper.
While unable to solve this question for any given fragment, we were able to reduce, following some known results, this question to a decidability question on the global of a fragment, a variety of categories.
Then decidability was obtained for many fragments, using different approaches.
They can be sorted in two cases.
The first case is when the global is understood and finitely describable.
Then we are able to decide a delay depending on the stability index and the said description.
The second case is when the fragment is expressive enough to handle the modular predicates.
This happens in particular if the fragment contains the local predicates and can use them extensively.

The main applications of these results are given in Figure~\ref{TableauFinal}, mainly on the levels of the quantifier alternation hierarchies, although this approach can be used on other fragments that satisfy the same hypotheses, such as the fragment $\FO[+1]$.

An interesting fact is that while the stability index often serves as a valid delay, this is still open whether this would hold for varieties of rank greater than two.

The question of solving the adding of modular predicate in a general setting seems achievable, although the more natural question would be to solve the decidability of the semidirect product by $\MODV$.
While we avoided this characterisation as it served no purpose in our approach, an aware reader could have noticed that Theorem~\ref{thm:derivedcategory}
proves that adding modular predicates is algebraically equivalent to a semidirect product by the length-multiplying variety of morphisms $\MODV$.
Then our question reduces to whether this semidirect product preserves decidability.
\cite{Auinger10} proved that the semidirect product in general does not preserve decidability, but the problem is still open for the case of $\MODV$.

\bibliographystyle{abbrvnat}
\bibliography{biblio}

\end{document}